\documentclass[12pt,fleqn,reqno]{amsart}

\usepackage{mathrsfs}
\usepackage{amsmath}
\usepackage{amssymb}
\usepackage{amsfonts}
\usepackage{amsthm}
\usepackage{a4wide}
\usepackage{graphicx}
\usepackage{color}
\usepackage[sans]{dsfont}
\usepackage{linearA}
\usepackage{pgfkeys}
\usepackage{longtable}
\usepackage{pdflscape}
\usepackage{float}
\usepackage{cancel}
\usepackage{dsfont}
\usepackage{setspace}
\usepackage{array}
\usepackage{hyperref}
\hypersetup{colorlinks=true,allcolors=[rgb]{0,0,0.6}}
\usepackage[english]{babel}
\usepackage[applemac]{inputenc}
\usepackage[T1]{fontenc}
\usepackage{tikz,pgflibraryshapes}
\usepackage{enumitem}
\usepackage[margin=1.37in]{geometry}

\newcommand{\vertiii}[1]{{\left\vert\kern-0.25ex\left\vert\kern-0.25ex\left\vert #1 
    \right\vert\kern-0.25ex\right\vert\kern-0.25ex\right\vert}}

\newtheorem{theorem}{Theorem}[section]
\newtheorem{lemma}[theorem]{Lemma}
\newtheorem{proposition}[theorem]{Proposition}
\newtheorem{corollary}[theorem]{Corollary}

\newtheorem{remark}[theorem]{Remark}

\numberwithin{equation}{section}

\usepackage[normalem]{ulem}
\usepackage{soul}
\setstcolor{blue}

\definecolor{green}{rgb}{0.0, 0.5, 0.5}

\definecolor{lgray}{gray}{0.9}
\definecolor{llgray}{gray}{0.95}
\definecolor{lllgray}{gray}{0.975}

\newcommand{\R}{\mathbb{R}}

\newcommand{\cE}{\mathcal{E}}
\newcommand{\cF}{\mathcal{F}}
  
\newcommand{\cG}{\mathcal{G}}

\newcommand{\nc}{\newcommand}

\nc{\al}{\alpha}
\nc{\del}{\delta}
\nc{\h}{\delta}
\nc{\G}{\Gamma}
\nc{\et}{\eta} 
\nc{\g}{\gamma}
\nc{\gam}{\gamma}
\nc{\ka}{\kappa}
\nc{\lam}{\lambda}
\nc{\Lam}{\Lambda}
\nc{\Om}{\Omega}
\nc{\om}{\omega}

\nc{\ta}{\tau}
\nc{\w}{\omega}
\nc{\io}{\iota}
\nc{\z}{\zeta}
\nc{\s}{\sigma}
\nc{\Si}{\Sigma}
\nc{\vphi}{\varphi}

\nc{\e}{\epsilon}

\nc{\ran}{\rangle}
\nc{\lan}{\langle}

\newcommand{\ra}{\rightarrow}

\newcommand{\ls}{\lesssim}
\newcommand{\gs}{\gtrsim}
\newcommand{\one}{\mathbf{1}}

\newcommand{\supp}{\operatorname{supp}}

\newcommand{\tr}{\mathrm{Tr}}

\newcommand{\dist}{\mathrm{dist}}

\nc{\bfone}{{\bf 1}}
\nc{\1}{{\bf 1}}

\newcommand{\p}{\partial}

\newcommand{\DETAILS}[1]{}

 \title[On Lieb-Robinson bounds for the Bose-Hubbard model]{On Lieb-Robinson bounds for \\ the Bose-Hubbard model}
 
  \author{J\'er\'emy Faupin}
\address{Institut Elie Cartan de Lorraine, Universit\'e de Lorraine, 57045 Metz Cedex 1, France}
 \email{jeremy.faupin@univ-lorraine.fr}
 
 \author{Marius Lemm}
 \address{Department of Mathematics, University of T\"ubingen, 72076 T\"ubingen, Germany}
 \email{marius.lemm@uni-tuebingen.de}
  
 \author{Israel Michael Sigal}
\address{Department of Mathematics, University of Toronto, 
Toronto, M5S 2E4, 
Canada}
 \email{im.sigal@utoronto.ca}

\begin{document}

\maketitle
\bibliographystyle{abbrv}
\begin{abstract}We consider the dynamics of the Bose-Hubbard model on general lattices and prove a Lieb-Robinson bound for observables whose supports are separated by an initially almost particle-free region. 
We further obtain a maximal velocity bound for particle transport through an initially empty region which also applies to long-range hopping.  Our techniques originate in the proofs of maximal velocity bounds for Schr\"odinger operators and scattering theory in non-relativistic QED. 
\end{abstract}

\section{Introduction}

Finiteness of the speed of  quantum transport 
(e.g.\ of particles and 
 perturbations) is at the root of our perception of the physical world. It is a central underpinning of the general theory that is also important  for the practical design of quantum devices. It is often posited 
 that the evolving quantum states   spread instantaneously. This   refers to the evolution of their supports.  
   Since   Quantum Mechanics is an intrinsically probabilistic theory, a natural question  
  to ask  is then whether information 
   propagates at finite speed with 
    high probability, rather than with probability $1$.  
  
    It was shown first by Lieb and Robinson (\cite{LR}) in Quantum Statistical Mechanics (QSM) 
   that quantum correlations propagate at most with finite speed  up to exponentially small errors (in $d-ct$, where $d$ is the distance to the source). This property is referred to as   the existence of an effective propagation cone (or, by analogy with relativistic theory, `light cone').

Lieb and Robinson's result stimulated considerable activity and led to a number of decisive applications in diverse fields, e.g., Hastings' famous proof of the area law for the entanglement entropy in gapped spin chains \cite{HastingsAL}; see also \cite{ElMaNayakYao, NachSim, NSY} 
 for reviews and further references. Originally proved for spin systems with bounded and finite range interactions,  the Lieb-Robinson bound was improved in \cite{Hastings,NachSim04}  
 and has by now been extended to infinite range interactions (see \cite{ElMaNayakYao, FossEtAl, MatKoNaka} and references therein),  to lattice oscillator models, which {in contrast to spin systems have an infinite-dimensional local Hilbert space (\cite{NachRazSchlSim, CrSerEis}), spin systems with dissipation (\cite{Pou,NVZ}) and, recently, to certain continuous models   (\cite{GebNachReSims}). 
For bosonic lattice gases, existence of the light cone and finiteness of the influence domain are more subtle and have been shown for special initial conditions in \cite{SHOE} and very recently in \cite{KuwSaito,YinLuc} (see also \cite{WaHazz})  for the key 
Bose-Hubbard model. See below for more detailed description of these results. 

Independently,  Sigal and Soffer (\cite{SigSof}) have shown that in Quantum Mechanics (QM) the `essential support' of the wave functions, i.e. the support up to negligible probability, spreads  with finite speed. 
The result of \cite{SigSof} was improved in \cite{Skib, HeSk, APSS} and extended to photons coupled to matter   in \cite{BoFauSig}. The approach of 
 \cite{SigSof, Skib, HeSk, APSS, BoFauSig} is  based on the method of differential inequalities for propagation observables and commutator expansions. It is fairly different from  the existing approaches in the literature on Lieb-Robinson bounds.
 
  In this paper, we bridge the gap 
   between these two independent lines of inquiry, in QSM and QM, and  we extend the QM approach of  \cite{SigSof, Skib, HeSk, APSS, BoFauSig} to QSM models  with unbounded interactions.     Specifically, combining with a new localization technique,  we establish the existence of an effective light cone for quantum transport for  the  Bose-Hubbard model through \textit{initially unoccupied regions,} that is, the particles are initially localized within spatially separated domains. Our results go beyond  the seminal earlier work \cite{SHOE}. 
   
We also discuss \textit{applications} of our bosonic Lieb-Robinson bound. For example, we show that the time evolution of initially localized observables (which in principle spreads across the entire lattice for any positive time $t>0$) can be approximated by purely local observables for short times. See Theorem \ref{thm:localapprox} for the precise statement. The analogous result for spin systems is central to prove an LPPL principle (``local perturbations perturb locally'') which in turn underpins the modern classification theory of topological quantum phases via Hastings spectral flow (also called quasi-adiabatic evolution) \cite{BMNS,Hastings}. Our results thus pave the way for developing the first bosonic theory of topological quantum phases. The bounds are accompanied by suitable particle number-dependent weights described below, a feature which is generally expected for bosonic systems.

As a second concrete application, we show how Theorem \ref{thm:localapprox} implies a bound on the information-theoretic task of \textit{quantum state transfer} in the Bose-Hubbard model. For this, we follow the line of argumentation developed in the context of quantum spin systems \cite{BHV,EppWhal} and replace the standard Lieb-Robinson bound by our bosonic version.
  
  The ideas developed in this paper also play an important role in our forthcoming paper \cite{FLS} which proves the first bound on the speed of macroscopic particle transport in the Bose-Hubbard model when starting from general initial states. While the techniques are related, the physical content of the two results are fundamentally different: The result proved here is a bona fide Lieb-Robinson bound which controls commutators of very general local observables and it requires a natural assumption on the initial distribution of particles. For comparison, the result in \cite{FLS} holds for general initial states but only controls macroscopic fractions of particles.

\subsection{Model and first main result}
We give the basic definitions and then state a special case of our main result, Theorem \ref{thm:MVE-SM} below. After discussing its interpretation and some applications, we provide a substantially more general result, Theorem \ref{thm:MVE-SM2}.

Let $\Lambda \subset \mathcal{L}$ be a subset of a lattice $\mathcal{L}\subset\mathbb{R}^d$. We use the Euclidean metrics and assume that nearest neighbors are separated by a distance $\ge1$.  We consider a system of bosons on $\Lambda$ described by the generalized Bose-Hubbard Hamiltonian
\begin{align}\label{HLam} H_\Lam &:= - \sum_{x\in\Lambda,y\in\Lambda} J_{xy} b_x^*b_y+\frac{g}2 \sum_{x\in\Lam}n_x(n_x-1)-\mu \sum_{x\in\Lam}n_x ,
\end{align} 
with $g>0$ and $\mu\in\R$. Here $b_x,b^*_x$ are the usual bosonic annihilation and creation operators, acting on the bosonic Fock space $\mathcal{F}$ over $\ell^2(\mathcal{L})$, and $n_x:=b^*_xb_x$ is the number operator at $x$. We will assume that there is an integer $p\ge 1$  such that
\begin{equation}\label{eq:kappa_J}
\kappa_J^{(p)} := \max_{x\in\Lambda}\sum_{y\in\Lambda}|J_{xy}||x-y|^{p} < \infty.
\end{equation}
The standard Bose-Hubbard Hamiltonian with nearest-neighbor hopping is $H_\Lambda$ with $J_{xy} = J \delta_{x\sim y}$ and $J\in\R$, where $x\sim y$ means that $x$ and $y$ are neighbors in $\Lambda$.

Let $N$ denote  the number operator, $N := \sum_{x\in\Lam} n_x$. As can be easily checked (see Appendix \ref{sec:SA},  Corollary \ref{cor:HN-com}), it commutes with $H_\Lambda$, which shows that  $H_\Lambda$ is self-adjoint as an infinite direct sum of self-adjoint bounded operators even in the case when $\Lambda$ is infinite, see Proposition \ref{prop:sa} below.

We write $\mathcal{D}(A)$ for the domain of an operator $A$.  To shorten notations, we write $H\equiv H_\Lambda$. We recall that  
the evolution of observables (self-adjoint bounded operators on the Fock space $\cF$)  is given by
\begin{equation*}
A_t:=\alpha_t(A)\equiv e^{-itH}Ae^{itH}.
\end{equation*}

One says that a (possibly unbounded) observable $A$ is \textit{localized} (or \textit{supported}) in a set $X\subset\Lambda$ iff $[A, b_x]=0$ for any $ x\in X^c:=\Lam\setminus X$.
For $x_0\in\Lambda$ and $b>0$, we use the notation
\begin{equation*}
\mathcal{B}(x_0,b):=\{ x \in \Lambda \, |\, |x-x_0|<b \}.
\end{equation*}
 Furthermore, we define the key quantity
\begin{equation}\label{nu}
\kappa:=\kappa_J^{(1)} \equiv \max_{x\in\Lambda}\sum_{y\in\Lambda}|J_{xy}||x-y|.
\end{equation}

Our main result provides an effective light cone for the propagation of information in the Bose-Hubbard model. The quantity $\kappa$ from above bounds the maximal speed of sound.

\begin{theorem}[Lieb-Robinson bound]\label{thm:MVE-SM}
Suppose that \eqref{eq:kappa_J} is satisfied for some integer $p\ge 1$. Let $c>\kappa$, $b>0$, $\delta_0>0$, $\rho>(1+\delta_0)b$ and $A,B$ be observables commuting with $N$ and localized in $\mathcal{B}(0,b)$ and $\mathcal{B}(0,2\rho-b)^c$, respectively. Let
 $\varphi\in\mathcal{D}(N^{\frac12})\subset\cF$ be such that
 \begin{equation}
n_x \varphi = 0,\qquad \forall x \in \Lambda, \textnormal{ with }b \le | x | \le 2\rho - b .\label{eq:hyp:empty}
\end{equation}
Then, for all $0\le t \le (2c)^{-1}(\rho-b)$,
 \begin{align} 
&\big|\big \langle \varphi , [ A_t, B] \varphi \big \rangle \big |  \le C t (\rho-b)^{1-p} \|A\|\|B\| \big\langle \varphi , N \varphi \big\rangle , \label{max-vel-estSM}
\end{align}
where $C$ is a positive constant depending on $p$, $c$ and $\delta_0$.
\end{theorem} 

For $p>2$,  \eqref{max-vel-estSM} shows that, with the probability approaching $1$, as $t\ra \infty$ an evolving observable, $A_t$, remains uncorrelated to (i.e. commuting with)  any other observable supported outside its light cone \[\{x\in \Lambda \,|\, \dist(x, \supp A)\le 2ct\},\] provided the supports of these observables are separated by initially unoccupied regions. In other words, it implies that the maximal speed of quantum propagation is bounded by the number $\kappa$ defined in \eqref{nu}. The assumptions on $A$, $B$ and $\varphi$ can be significantly relaxed as we explain in Theorem \ref{thm:MVE-SM2} below.

Theorem \ref{thm:MVE-SM} is proven in Section \ref{sec:thm2}.

\subsection{Discussion and interpretation}

The bound \eqref{max-vel-estSM} (and \eqref{max-vel-estSM2} below) impose a direct constraint on the propagation of information through the quantum channel defined by the time evolution $\al_{-t}$ of quantum states via the Bose-Hubbard Hamiltonian. For example, (see e.g. \cite{Pou}), assume that Alice at a location $X$ is in possession of a state $\rho$ and an observable $A$ and would like to send a signal through the quantum channel $\alpha_{-t}$ to Bob who is at $Y$ and who possesses the same state $\rho$ and an observable $B$. To send the message ``$1$'', Alice uses $A$ as a Hamiltonian to evolve $\rho$ for a short time, $r$, and then applies the quantum channel, sending Bob the state 
$\alpha_{-t}(\rho_r)=\alpha_{-t}(e^{-iAr}\rho e^{iAr})$. To send the message ``$0$'', Alice simply sends $\alpha_{-t}(\rho)$. To see whether Alice sent ``$0$'' or ``$1$'', Bob computes the difference between the expectations of $B$ in the states $\al_{-t}(\rho_r)$ and $\al_{-t}(\rho)$. Using the approximation $\rho_r=e^{-iAr}\rho e^{iAr}\approx \rho - r i[A, \rho]$, this gives
\begin{align} \label{CommEstPhys}\tr\big[B\al_{-t}(\rho_r)-B\al_{-t}(\rho) \big]&\approx  r\tr\big[B\al_{-t}( i[\rho, A])\big]\notag\\
&= r\tr\big( i[A,\al_{t}(B)] \rho\big).\end{align} 
Taking $\rho=|\varphi\ran\lan \varphi|$, dividing by $r$ and swapping $A$ and $B$ gives  the expression estimated in  Eq. \eqref{max-vel-estSM} or \eqref{max-vel-estSM2}.

Generally speaking, the usefulness of Theorem \ref{thm:MVE-SM} is that one can derive from it all the consequences of standard Lieb-Robinson bounds for the first time in a bosonic context, \textit{provided} that one is interested in (a) bounding expectations instead of norms and (b) restricting to states with some particle-free regions in the sense of \eqref{eq:hyp:empty} (or almost particle-free as described in the generalization below). We emphasize that some compromises along these lines are clearly necessary in the bosonic context.

Our proof of Theorem \ref{thm:MVE-SM} also shows the following result which is of interest in its own right.
Note that the support of an initially localized observable generally spreads over the entire lattice immediately for any $t>0$. Nonetheless, we show that local observables can be approximated by local observables for sufficiently short times. 

\begin{theorem}\label{thm:localapprox}
Under the assumptions of Theorem \ref{thm:MVE-SM}, there exists a local approximation $[A_t]_\rho$ to the time-evolved observable $A_t$ such that $[A_t]_\rho$ is localized in $\mathcal B(0,\rho)$ and, 
\begin{equation}\label{eq:loc}
\big|\big \langle \varphi ,\left(A_t-[A_t]_\rho\right)\varphi \big \rangle \big |  \le C t (\rho-b)^{1-p} \|A\| \big\langle \varphi , N \varphi \big\rangle , 
\end{equation}
for all $0\le t \le (2c)^{-1}(\rho-b)$, where $C$ is as in Theorem \ref{thm:MVE-SM}.
\end{theorem}

The local approximation can be defined explicitly in terms of the time-evolved observable $e^{-itH_{\mathcal{B}(0,\rho)}}Ae^{itH_{\mathcal{B}(0,\rho)}}$, where $H_{\mathcal{B}(0,\rho)}$ is defined as in \eqref{HLam} (see \eqref{eq:def_alpha_trho} for the precise expression of $[A_t]_\rho$). This does not require the introduction of a partial trace, which can be subtle in the infinite-dimensional setting; see also \cite[Lemma 3.2]{BMNS}. The assumptions on $A$ and $\varphi$ can be significantly relaxed, see Theorem \ref{thm:MVE-SM2} below. Theorem \ref{thm:localapprox} is proven in Section \ref{sec:thm2}.

We briefly review the literature. The only Bose system where Lieb-Robinson bounds are available so far (for special initial states) is the one described by  the Bose-Hubbard model.  For initial conditions where all particles are located in a \textit{bounded} set, the existence of the light cone was shown in \cite{SHOE}. In other words, the result of \cite{SHOE} captures transport of particles \textit{arriving} in an empty region. For comparison, Theorem \ref{thm:MVE-SM} also captures transport \textit{through} an empty region.

 In the recent work \cite{WaHazz}, the Lieb-Robinson bound is proven for  the truncated  Bose-Hubbard model with  the Lieb-Robinson speed of the order $\mathcal{O}(\sqrt{\bar N})$, where $\bar N$ is the average number of particles. \cite{KuwSaito} proved that   local excitations of static solutions satisfying certain stringent low-boson-density conditions stay within $\mathcal{O}(t \log^2 t)$-neigbourhood of their initial support.  The recent work \cite{YinLuc} derives a light cone for  $\mathrm{Tr}(e^{-\mu N} [A_t,B])$, i.e., for the special initial states $e^{-\mu N}$.

\begin{remark}\normalfont
\begin{enumerate}[label=(\roman*)]

\item For the standard Bose-Hubbard Hamiltonian with nearest-neighbor hopping $J_{xy} = J \delta_{x\sim y}$, $J>0$, the maximal velocity is given by
\begin{equation*}
\kappa = J \max_{x\in\Lambda} \# \{y\in\Lambda\, | \, x\sim y\}.
\end{equation*}
Moreover \eqref{eq:kappa_J} is satisfied for all $p\ge 1$ and hence \eqref{max-vel-estSM} and \eqref{max-vel-estSM2} hold for all $p\geq 1$.
\item  Theorem  \ref{thm:MVE-SM} implies that the maximal speed of propagation is bounded by the number $\kappa$ defined in \eqref{nu}. Though this result looks natural, the result relies on the implicit energy cut-off baked into the lattice step (as do the Lieb-Robinson bounds published elsewhere). If we introduce a variable lattice step $h$, then the maximal speed would blow up as $h\ra 0$.

One can display the dependence on the energy cut-off explicitly by considering initial conditions of the form $\varphi=g(H_\Lam)\psi$, with say $g\in \mathrm{C}_0^\infty(\R)$ and $\psi\in \cF$. Such an energy cut-off is also necessary in the continuous case: without it, particles and information propagate with infinite speed (cf.  \cite{SigSof, Skib, HeSk, APSS, BoFauSig}). 
\end{enumerate}
\end{remark}

\subsection{Generalizations and further discussion}
Theorems \ref{thm:MVE-SM} and \ref{thm:localapprox} generalize to unbounded observables $A$ and $B$ and initial states $\varphi$ with a small number of particles between the supports of $A$ and $B$. For $b_1,b_2\ge0$, we set
\begin{equation*}
\mathcal{C}_{b_1,b_2}:=\{ x \in \Lambda \,|\, b_1\le|x|\le b_2 \}.
\end{equation*}
We denote by $\mathbb{N}$ the set of positive integers and $\mathbb{N}_0=\mathbb{N}\cup\{0\}$. We consider the following subset of unbounded observables: We say that an unbounded self-adjoint operator $A$ on the bosonic Fock space has a degree at most $\nu_A\in\mathbb{N}_0$ if, for all $n\in\mathbb{N}_0$,
\begin{equation}
\|A\|_{n}:=\big\| (N+1)^{\frac{n}{2}} A (N+1)^{-\frac{n+\nu_A}{2}} \big\|<\infty. \label{eq:compl}
\end{equation}
Note that \eqref{eq:compl} holds if $A$ is polynomial of degree $\le\nu_A$ in $b_x$, $b_x^*$. Since $H$ commutes with $N$, given an observable of degree at most $\nu_A$, the evolution $A_t$ is well-defined on $\mathcal{D}(N^{\nu_A/2})$. To simplify formulas below, we also set
\begin{equation*}
\vertiii{A}_{\nu}:=\max_{0\le n\le\nu+1}\|A\|_n.
\end{equation*}

\begin{theorem}[Lieb-Robinson bound for general observables and states]\label{thm:MVE-SM2}
Suppose that \eqref{eq:kappa_J} is satisfied for some integer $p\ge1$. Let $c>\kappa$, $b>0$, $\delta_0>0$, $\rho>(1+\delta_0)b$, and $A,B$ be self-adjoint observables of degrees at most $\nu_A,\nu_B\in\mathbb{N}_0$, such that $A$ and $B$ are localized in $\mathcal{B}(0,b)$ and $\mathcal{B}(0,2\rho-b)^c$, respectively. For all
 $\varphi\in\mathcal{D}(N^{\frac12(1+\nu_A+\nu_B)})\subset\cF$, we have
 \begin{align} 
&\big|\big \langle \varphi , [ A_t, B] \varphi \big \rangle \big |  \le  C t (\rho-b)^{1-p} \vertiii{A}_{\nu_B} \vertiii{B}_{\nu_A} \langle \varphi , M_\rho \varphi \rangle ,\label{max-vel-estSM2}
\end{align}
for all $0\le t \le (2c)^{-1}(\rho-b)$, where 
\begin{align*}
\quad M_\rho := (N+1)^{\nu_A+\nu_B}\Big( N + \sum_{x\in\mathcal{C}_{b,2\rho-b}} m^{p-1}_\rho(x)  n_x+1\Big) ,
\end{align*}
 $m_\rho(x):=\min(|x|,2\rho-|x|)$ and $C$ is a positive constant depending on $p$, $c$ and $\delta_0$. 
  Moreover, there exists a local approximation $[A_t]_\rho$ to the time-evolved observable $A_t$ such that  $[A_t]_\rho$ is localized in $\mathcal B(0,\rho)$ and
\begin{align}
\big|\big \langle \varphi ,\left(A_t-[A_t]_\rho\right)\varphi \big \rangle \big | 
\le C t (\rho-b)^{1-p} \vertiii{A}_0 \big\langle \varphi , M_\rho \varphi \big\rangle . \label{eq:general_loc}
\end{align}
\end{theorem}

\begin{remark}\normalfont\label{rmk:extensions}
\begin{enumerate}[label=(\roman*)]
\item By polarization, Theorems \ref{thm:MVE-SM} and \ref{thm:MVE-SM2} imply analogous estimates on off-diagonal matrix elements, i.e., $\big \langle \varphi , [ A_t, B] \tilde\varphi \big \rangle$ with $\varphi\neq \tilde\varphi$. In turn, \eqref{max-vel-estSM2} is equivalent to the weighted operator norm bound
\begin{equation*}
\big\|M_\rho^{-\frac12} [ A_t, B] M_\rho^{-\frac12}\big \|  \le  C t (\rho-b)^{1-p} \vertiii{A}_{\nu_B} \vertiii{B}_{\nu_A} ,
\end{equation*}
 which also gives the estimate
\begin{equation*}
 \mathrm{Tr} \big ( [ A_t, B] \gamma \big) \le  C t (\rho-b)^{2-p} \vertiii{A}_{\nu_B} \vertiii{B}_{\nu_A} \mathrm{Tr}\big( M_\rho^{\frac12}\gamma M_\rho^{\frac12} \big),
\end{equation*}
for any trace class operator $\gamma$ such that $\mathrm{Tr}\big( M_\rho^{\frac12}\gamma M_\rho^{\frac12} \big)<\infty$. 

\item By approximate translation invariance, the balls  $\mathcal{B}(0,b)$ and $\mathcal{B}(0,2\rho-b)$ in Theorems \ref{thm:MVE-SM} and \ref{thm:MVE-SM2} can be replaced by $\mathcal{B}(z,b)$ and $\mathcal{B}(z,2\rho-b)$, for any $z\in\Lambda$. 

\item Theorems \ref{thm:MVE-SM} and \ref{thm:MVE-SM2} also holds for non-self-adjoint observables $A$, $B$, replacing the definition \eqref{eq:compl} of $\|A\|_{n}$ by
\begin{equation}
\|A\|_{n}:=\max\big(\big\| (N+1)^{\frac{n}{2}} A (N+1)^{-\frac{n+\nu}{2}} \big\|,\big\| (N+1)^{\frac{n}{2}} A^* (N+1)^{-\frac{n+\nu}{2}} \big\|\big).
\end{equation}
Indeed, by linearity, it suffices to decompose $A=\mathrm{Re}(A)+i\mathrm{Im}(A)$, likewise for $B$, and then apply Theorem \ref{thm:MVE-SM} or.\ref{thm:MVE-SM2}
\item Our result in Theorem \ref{thm:MVE-SM} and \ref{thm:MVE-SM2} hold for times $t$ such that $0\le t < (2c)^{-1}(\rho-b)$. As will follow from our proof, we also have the following estimate, which holds for $t<c^{-1}(\rho-b)$, but with a worse decay rate: for $p\ge1$, 
 \begin{align} 
&\big|\big \langle \varphi , [ A_t, B] \varphi \big \rangle \big |  \le C t (\rho-b)^{\frac{1-p}{2}} \vertiii{A}_{\nu_B}\vertiii{B}_{\nu_A} \langle \varphi,M_\rho\varphi\rangle. \label{max-vel-estSM_2nd}
\end{align}
\item Our proof will show that the expectation value $\langle\varphi,M_\rho\varphi\rangle$ in the right-hand-side of \eqref{max-vel-estSM2} can be replaced by the smaller term
\begin{equation*}
\big\| M_\rho^{\frac12}\varphi\big\| \big\| (M_\rho^{(0)})^{\frac12}\varphi\big\|, \quad\text{with}\quad M_\rho^{(0)} := N + \sum_{x\in\mathcal{C}_{b,2\rho-b}} m^{p-1}_\rho(x)  n_x+1 .
\end{equation*}
\end{enumerate}
\end{remark}

Note that at the quantum energies in nature and laboratories (besides particle accelerators), the maximal speed of propagation implied by our results is much below the speed of light, so the non-relativistic nature of Quantum Mechanics is unimportant here. 

The proof of Theorem \ref{thm:MVE-SM2} is given in Appendix \ref{sec:thm3}.
  
 \subsection{Application: Bound on quantum state transfer}
 We combine our results with information-theoretic techniques from \cite{BHV,EppWhal} to derive a bound on the information-theoretic task of state transfer. 
 
 We recall that \textit{quantum state transfer} describes the task of transfering a quantum state $\gamma$ from a region $X$ to another region $Y$ by applying (i) a state-preparation unitary operator $A$ on a region $X$ and (ii) the Heisenberg time evolution of the whole system, in our case the Bose-Hubbard model.
 
 Following \cite{EppWhal}, we use the figure of merit for the state transfer
\begin{equation}\label{eq:fom}
 F(\mathrm{Tr}_{Y^c} \alpha_t(\gamma),\mathrm{Tr}_{Y^c} \alpha_t(A\gamma A^*))
\end{equation}
where $F(\rho,\sigma)=\|\sqrt{\rho}\sqrt{\sigma}\|_{\mathfrak S^1}$ is the fidelity with $\|\cdot\|_{\mathfrak S^1}$ denoting the trace norm (also called Schatten-$1$-norm). Note that the fidelity between two quantum states (density matrices) equals $1$ if and only $\rho=\sigma$. As explained in \cite{EppWhal}, the availability to effect quantum state transfer on general input states, in particular orthogonal states, requires the fidelity in \eqref{eq:fom} to be small.

Our result in this setting is the following lower bound on  the fidelity of quantum state transfer in \eqref{eq:fom} when $\gamma$ is pure and the time is short compared to the transfer distance.

\begin{corollary}[Quantum state transfer bound]\label{cor:qst}
Let $X=\mathcal B(0,b)$ and $Y=\mathcal B(0,2\rho-b)$. Under the same assumptions as in Theorem \ref{thm:MVE-SM2}, with $A$ a unitary operator localized in $X$, we have with $\gamma=\vert\phi\rangle\langle\phi\vert$,
\begin{equation}\label{eq:qst}
\begin{aligned}
&
F(\mathrm{Tr}_{Y^c} \alpha_t(\gamma),\mathrm{Tr}_{Y^c} \alpha_t(A\gamma A^*))\\
&\geq 1- C t (\rho-b)^{2-p}\vertiii{A}_0
\sup_{0\le t \le (2c)^{-1}(\rho-b)}\big\|M_{\rho}^{\frac12} \alpha_t(\gamma)M_\rho^{\frac12}\big\|_{\mathfrak S^1}.
\end{aligned}
\end{equation}
\end{corollary}

As noted above, this corollary establishes a limit on the best-possible quantum state transfer protocols for the Bose-Hubbard model. Note that 
\begin{equation*}
\big\|M_{\rho}^{\frac12} \alpha_t(\gamma)M_\rho^{\frac12}\big\|_{\mathfrak S^1}=\langle \phi_t, M_\rho \phi_t\rangle
\end{equation*}
 where $ \alpha_t(\gamma)=\vert\phi_t\rangle\langle\phi_t\vert$. The proof of Corollary \ref{cor:qst} is presented in Appendix \ref{app:qst}.

 \subsection{Further dynamical bounds on particle transport}
The proof of Theorems \ref{thm:MVE-SM} and \ref{thm:MVE-SM2} are based on a 
 Fock space localization technique    together with a special case of the following theorem.
   
   We denote by $\chi_S$ the characteristic function of a set $S\subset\Lambda$. Recall that the second quantization of a one-particle operator $a$ on $\ell^2(\Lambda)$ with operator kernel $a_{x,y}$ is given by $\mathrm{d}\Gamma(a):=\sum_{x,y} a_{x,y} b_x^* b_y$. Abusing notations, a function $F:\Lambda\to\mathbb{C}$ is identified with the multiplication operator that acts diagonally on $\ell^2(\Lambda)$ as $Ff(x)=F(x)f(x)$. Hence
   \begin{equation*}
   	\mathrm{d}\Gamma(F)=\sum_{x\in\Lambda}F(x)b_x^*b_x.
   \end{equation*}
 If $F=\chi_S$ with $S\subset\Lambda$, we also set
 \begin{equation*}
 N_S:=\mathrm{d}\Gamma(\chi_S) = \sum_{x\in S}n_x.
 \end{equation*}
 For any initial state $\psi_0\in\mathcal{F}$, we denote  by $\psi_t:=e^{-itH}\psi_0$ the solution to the Schr\"odinger equation $i\partial_t\psi_t=H\psi_t$.
\begin{theorem}[Particle propagation bound]\label{thm:max-vel-est}
Suppose that \eqref{eq:kappa_J} is satisfied for some integer $p \ge 1$. For all $c > \kappa$, $\delta_0>0$ and all integers $n\le p-1$, there exists $C>0$ such that, for all $b>0$, $\rho > b+\delta_0$ and $\psi_0\in\mathcal{D}(N^{\frac12})\subset\mathcal{F}$,
\begin{align}
&\sup_{0\le t<c^{-1}(\rho-b)} \big \langle \psi_t, N_{|x|> \rho} \psi_t \big \rangle \notag \\
& \quad\le \big(1+C(\rho-b)^{-1}\big) \big\langle\psi_0,N_{|x|>b} \psi_0\big\rangle +  C (\rho-b)^{-n}  \langle\psi_0,N\psi_0\rangle,\label{max-vel-est} 
\end{align}
and
\begin{align}
&\sup_{0\le t<c^{-1}(\rho-b)}\big \langle \psi_t, N_{|x|< b} \psi_t \big \rangle \notag \\
&\quad\le \big(1+C(\rho-b)^{-1}\big) \big\langle\psi_0,N_{|x|<\rho}\psi_0\big\rangle +  C (\rho-b)^{-n}  \langle\psi_0,N\psi_0\rangle. \label{max-vel-est_2}
\end{align}
\end{theorem}
Equation \eqref{max-vel-est} shows that the expectation of the number of particles in the region $\{|x|>\rho\}$ in the evolved state $e^{-itH}\psi_0$ does not exceed the number of particles initially in the region $\{|x|>b\}$, up to small remainder terms. In other words, the probability that particles are transported from $\{|x|\le b\}$ to $\{|x|>\rho\}$ is small for all times $t$ satisfying $b+ct<\rho$. 

Note that (a) by translation invariance, $x$ can be replaced by $x-z$ in \eqref{ex:max-vel-est_3} and \eqref{ex:max-vel-est_4}, for any $z\in\Lambda$,  and (b) Theorem \ref{thm:max-vel-est} implies the following estimates  (cf.  \cite{SHOE})
\begin{align}\label{ex:max-vel-est_3}
&\Big \| (N_{|x|> \rho})^{\frac12} e^{-itH}\Gamma(\chi_{|x|<b}) \psi_0 \Big \| \ls  (\rho-b)^{-n} \big \| N^{\frac12}\psi_0 \big\|, \quad \rho>b+ct,\\
 \label{ex:max-vel-est_4}
&\Big \| (N_{|x|< b})^{\frac12} e^{-itH}\Gamma(\chi_{|x|>\rho}) \psi_0 \Big \| \ls  (\rho-b)^{-n} \big \| N^{\frac12}\psi_0 \big\|, \quad \rho>b+ct,
\end{align}
where $\Gamma(a)$ denotes the operator on $\mathcal{F}$ defined by its restriction to the $n$-particle space $\mathcal{F}_n$ by $\Gamma(a)|_{\mathcal{F}_n}=\otimes^n a$, $\Gamma(a)|_{\mathcal{F}_0}=\one_{\mathcal{F}_0}$.
Equation \eqref{ex:max-vel-est_3} shows that,  if the initial state $\psi_0$ is localized in $\{|x|< b\}$, then 
 the probability that particles are transported from $\{|x|\le b\}$ to $\{|x|>\rho\}$ in time $t\le \frac1c (\rho-b)$ is $\le C(\rho-b)^{-n}$.

Theorem \ref{thm:max-vel-est} is  proven in Section \ref{sec:thm1}. 
The idea of the proof of Theorem \ref{thm:max-vel-est}  
  is as follows. Let $\Phi(t)$ a positive, differentiable 
   family of observables and denote $\lan A\ran_t :=\lan\psi_t, A\psi_t\ran$. Note the relation  
\begin{align}\label{dt-Heis}
&{d\over{dt}}\left<\Phi(t)\right>_t =\lan D\Phi(t)\ran_t,\ \text{ where }\ 
 D\Phi(t)=i[H,\Phi(t)]+{\partial\over{\partial t}}\Phi(t).
\end{align}
We call $D$ the {\it Heisenberg derivative}.  
Using $\lan \Phi(t)\ran_t= \lan \Phi(0)\ran_0+\int_0^t \p_r\left<\Phi(r)\right>_r dr$, we find 
\begin{align} \label{eq-basic}  
\lan \Phi(t)\ran_t-\int_0^t \lan D\Phi(r)\ran_r dr= \lan \Phi(0)\ran_0,
\end{align}
which we call the {\it basic equality}.
If (a) $\Phi(t)\ge 0$ and $ \lan D\Phi(t)\ran_t \le 0$, for a certain class of initial conditions, modulo fast time-decaying terms,
then relation \eqref{eq-basic} gives estimates on the positive terms $ \lan \Phi(t)\ran_t$ and $-\int_0^t \lan D\Phi(r)\ran_r dr$. If, in addition, (b) $ \Phi(t) \gs \mathrm{d}\Gamma(\chi_{|x|> \rho})$, modulo fast time-decaying terms,
then, we have an estimate on $\lan \mathrm{d}\Gamma(\chi_{|x|> \rho})\ran_t$ leading to Theorem \ref{thm:max-vel-est}. 
So our goal is to find a  family, $\Phi(t)$,  of observables, called propagation observables, satisfying conditions (a) and (b).

\section{Proof of Theorem \ref{thm:max-vel-est}}\label{sec:thm1}

\subsection{Differential inequalities}
We fix $c> \kappa$, $v:=\frac12(c+\kappa)$ and  let $\cE$ be the set of functions  $0\le f\in \mathrm{C}^\infty(\R)$, supported in $\R^+=(0,\infty)$ and
satisfying $f(\lam)=1$ for $\lam\ge c-v$,  and $f^\prime\ge 0$, with $\sqrt{f'}\in \mathrm{C}^\infty(\mathbb{R})$. 
 
To shorten formulas below, we will use the following notations: 
\begin{align*}
 |x|_{ts} :=s^{-1}(|x| -b -v t), 
  \end{align*}
and, for a bounded function $f:\mathbb{R}\to\mathbb{R}$,
\begin{equation*}
N_{f,ts} := \mathrm{d}\Gamma\big( f(|x|_{ts})\big).
\end{equation*}
Recall that, given an operator $A$ and an initial state $\psi_0$, we denote
\begin{equation*}
\langle A \rangle_t := \langle \psi_t , A \psi_t \rangle, \quad \psi_t :=e^{-iHt} \psi_0.
\end{equation*}

Next is a key statement in the proof of Theorem \ref{thm:max-vel-est}:

\begin{proposition}\label{prop:propag-est1} 
Suppose that \eqref{eq:kappa_J} is satisfied for some integer $p \ge 1$. For all $c>\kappa$, $f\in \cE$ and any integer $n\le p -1$, there are $j_k\in \cE$, $2\le k\le n$, and $C>0$ such that, for all $b>0$, $t>0$ and $s>0$, 
\begin{align}
&\int_0^t \big\lan N_{f',rs} \big\ran_r dr  \le C \Big(
  s\big \lan N_{f,0s} \big\ran_0 
+ \sum_{k=2}^n s^{-k+2} \big\lan N_{j_k,0s} \big\ran_0 +  t s^{-n}\langle N\rangle_0\Big), \label{propag-est31} 
\end{align}
where the sum should be dropped if $n=0,1$.
\end{proposition} 

 We will use the following easy lemma whose proof is postponed to Appendix \ref{app:lemma}. 
\begin{lemma}\label{lm:taylor}
Let $f\in\cE$. For all $n\in\mathbb{N}$, there exist $\tilde f_k \in \cE$, $2\le k\le n$ and positive constants $C_{f,k}$ such that, for all $x,y\in\R$,  and with $u:=(f')^{1/2}$ and $\tilde u_k:=(\tilde f'_k)^{1/2}$,
\begin{align*}
f(x)-f(y) = (x-y)u(x)u(y) + \sum_{k=2}^n (x-y)^kh_k(x,y) + \mathcal{O}\big((x-y)^{n+1}\big) ,
\end{align*}
where the sum should be dropped for $n=1$ and, for $2\le k\le n$,
\begin{align*}
|h_k(x,y)| \le C_{f,k}  \tilde u_k(x) \tilde u_k(y).
\end{align*}
\end{lemma}

\begin{proof}[Proof of Proposition \ref{prop:propag-est1}]
For $n=0$, the proposition is obvious, since $N_{f',rs}\le C_fN$. In the following we fix $n\in\mathbb{N}$.

 We use  the time-dependent observable 
\begin{align}\label{propag-obs1}
\Phi_s(t) & = N_{f,ts}, \quad  f \in \cE, 
\end{align}
with $t,s>0$.  In order to estimate $\left<\Phi_s(t)\right>_t=\lan\psi_t,\,\Phi_s(t)\psi_t\ran$,  
we apply \eqref{dt-Heis} and  the basic equality \eqref{eq-basic}. 
We start by computing $D\Phi_s(t)$. First, we have
\begin{align} \label{eq:deriv}
{\partial\over{\partial t}}\Phi_s(t)=-s^{-1}v \, N_{f',ts}.
\end{align}
Then using Lemma \ref{lem:HRf-com}, we have
\begin{align}
i\big [H , \Phi_s(t) \big ] &= \sum_{x,y\in\Lambda, x\neq y} J_{xy} \{f(|x|_{ts})-f(|y|_{ts}) \}b_x^*b_y , \label{eq:commut}
\end{align}
in the sense of quadratic forms on $\mathcal{D}(H)\cap\mathcal{D}(N)$. Applying Lemma \ref{lm:taylor} gives
\begin{align*}
f(|x|_{ts})-f(|y|_{ts}) &= s^{-1}(|x|-|y|)u(|x|_{ts})u(|y|_{ts}) \\
+ & \sum_{k=2}^n \frac{(|x|-|y|)^k}{s^k} h_k(|x|_{ts},|y|_{ts}) + \mathcal{O}\big(s^{-n-1}(|x|-|y|)^{n+1}\big) ,
\end{align*}
where the sum should be dropped if $n=1$ and, for $2\le k\le n$,
\begin{align*}
|h_k(x,y)| \le C_{f,k}  \tilde u_k(x) \tilde u_k(y).
\end{align*}
Here we have set $u:=(f')^{1/2}$ and $\tilde u_k:=(\tilde f'_k)^{1/2}$, with $\tilde f_k\in\cE$. Inserting this into \eqref{eq:commut} yields:
\begin{align}
i\big [H , \Phi_s(t) \big ] &= s^{-1} \sum_{x,y\in\Lambda, x\neq y} J_{xy} (|x|-|y|)u(|x|_{ts})u(|y|_{ts}) b_x^*b_y \notag \\
&\quad + \sum_{k=2}^n \sum_{x,y\in\Lambda, x\neq y} J_{xy} \frac{(|x|-|y|)^k}{s^k} h_k(|x|_{ts},|y|_{ts}) b_x^*b_y \notag \\
&\quad + s^{-n-1} \sum_{x,y\in\Lambda, x\neq y} J_{xy} \mathcal{O}\big((|x|-|y|)^{n+1}\big). \label{eq:commut_1}
\end{align}
Using the Cauchy-Schwarz inequality and the fact that $u^2=f'$, we deduce the following form inequalities for the first term: for all $\varphi\in\mathcal{D}(N^{\frac12})$,
\begin{align*}
&\Big |\Big \langle \varphi , \sum_{x,y\in\Lambda, x\neq y} J_{xy} (|x|-|y|)u(|x|_{ts})u(|y|_{ts}) b_x^*b_y \varphi \Big \rangle \Big |\\
&\le\sum_{x,y\in\Lambda, x\neq y} | J_{xy}||x-y| \big|\big \langle u(|x|_{ts}) b_x \varphi , u(|y|_{ts}) b_y \varphi \big \rangle\big| \\
&\le \Big( \sum_{x\in\Lambda} f'(|x|_{ts}) \langle\varphi,b^*_xb_x\varphi\rangle\Big(\sum_{y\in\Lambda,y\neq x} |J_{xy}||x-y|\Big)\Big)^{\frac12} \\
&\quad\times \Big(\sum_{y\in\Lambda} f'(|y|_{ts}) \langle\varphi,b^*_yb_y\varphi\rangle\Big(\sum_{x\in\Lambda,x\neq y} |J_{xy}||x-y|\Big)\Big)^{\frac12} \\
&\le\kappa\big\langle\varphi,N_{f',ts} \,\varphi\big\rangle.
\end{align*}
Higher-order terms can be treated in the same way, yielding, for all $2\le k\le n$ and $\varphi\in\mathcal{D}(N^{\frac12})$,
\begin{align*}
&\Big |\Big \langle \varphi , \sum_{x,y\in\Lambda, x\neq y} J_{xy} (|x|-|y|)^k h_k(|x|_{ts},|y|_{ts}) b_x^*b_y \varphi \Big \rangle \Big |\\
&\le\kappa^{(k)}_JC_{f,k}\big\langle\varphi,N_{\tilde f^\prime_k,ts} \,\varphi\big\rangle ,
\end{align*}
since $n\le p $. Likewise, the remainder term in \eqref{eq:commut_1} can be estimated as
\begin{align*}
&\Big |\Big \langle \varphi , \sum_{x,y\in\Lambda, x\neq y} J_{xy} \mathcal{O}\big((|x|-|y|)^{n+1}\big) b_x^*b_y \varphi \Big \rangle \Big |\le \kappa_J^{(n+1)}C_{f,n} \big\langle\varphi,N\varphi\big\rangle,
\end{align*}
since $n+1\le p $.

Putting together the previous inequalities gives
\begin{align*}
 i\big [H , \Phi_s(t) \big ] &\le  \kappa s^{-1} N_{f',ts} + \sum_{k=2}^n \kappa^{(k)}_J C_{f,k} s^{-k} N_{\tilde f_k^\prime ,ts}   + C_{f,n} \kappa_J^{(n+1)} s^{-n-1} N,
\end{align*}
in the sense of quadratic forms on $\mathcal{D}(H)\cap\mathcal{D}(N)$. Combining this estimate with \eqref{eq:deriv}, we arrive at 
\begin{align*}
 D\Phi_s(t)  &\le (\kappa-v) s^{-1} N_{f',ts} + \sum_{k=2}^n \kappa^{(k)}_J C_{f,k} s^{-k} N_{\tilde f_k^\prime ,ts}   + C_{f,n} \kappa_J^{(n+1)} s^{-n-1} N.
\end{align*}

Applying the previous inequalities to the vector 
$\psi_t=e^{-iHt}\psi_0$, with $\psi_0\in\mathcal{D}(H)\cap\mathcal{D}(N)$, using Eqs. \eqref{dt-Heis} and \eqref{eq-basic} and the definition $\Phi_s(t)= N_{f,ts}$ give  
\begin{align} \label{propag-est2} 
&\big\lan N_{f,ts}\big \ran_t+(v-\kappa)s^{-1}\int_0^t\big \lan N_{f',rs} \big\ran_r dr\notag\\
& \le \,\big \lan N_{f,0s}\big \ran_0+C_{f,n} \sum_{k=2}^n s^{-k}\int_0^t \big\lan N_{\tilde f'_k,rs}\big \ran_r dr + C_{f,n} t s^{-n-1}\langle N\rangle_0. 
\end{align}
Since $\kappa < v$, \eqref{propag-est2} implies (after dropping $\lan N_{f,ts} \ran_t$ and multiplying by $s(v-\kappa)^{-1}$) that
\begin{align}
\int_0^t \big \lan N_{f',rs} \big\ran_r dr\, 
 \le C_{f,c,n} \Big( \, &
  s \big\lan N_{f,0s} \big\ran_0 + \sum_{k=2}^n s^{-k+1} \int_0^t \big\lan N_{\tilde f'_k,rs} \big\ran_r dr \notag\\
  &+  t s^{-n}\langle N\rangle_0\Big). \label{propag-est3} 
\end{align}
If $n=1$, the sum should be dropped, which gives estimate \eqref{propag-est31}. If $n\ge2$, applying \eqref{propag-est3} to the term $\int_0^t \lan N_{\tilde f'_2,rs}  \ran_r dr$ and using that if $f_1,f_2\in\cE$ then $f_1+f_2\lesssim f_3$ for some $f_3\in\cE$, we obtain
\begin{align}
\int_0^t\big \lan  N_{f',rs}\big \ran_r dr\, 
 \le C_{f,c,n} \Big( \, &
  s \big\lan  N_{f,0s} \big\ran_0 + \big\lan N_{\tilde f_2 , 0s} \big\ran_0 + \sum_{k=3}^n s^{-k+1} \int_0^t \big\lan N_{\tilde{\tilde f}'_k,rs}\big \ran_r dr \notag \\
  & +  t s^{-n}\langle N\rangle_0\Big), \label{propag-est33} 
\end{align}
for some $\tilde{\tilde{f}}_k\in\cE$. Repeating the procedure, we arrive at \eqref{propag-est31} for $\psi_0\in\mathcal{D}(H)\cap\mathcal{D}(N)$. By a standard density argument, this extends to $\psi_0\in\mathcal{D}(N^{1/2})$ and hence Proposition \ref{prop:propag-est1} is proven.
\end{proof}

\subsection{Concluding the proof of Theorem \ref{thm:max-vel-est}}
\begin{proof}[End of the proof of \eqref{max-vel-est}]
Recall that $\supp f\subset \R^+$ for any $f\in \cE$, 
and hence, for any $s>0$,
\begin{equation*}
\supp f\Big(\frac{\cdot-b}{s}\Big) \subset (b+\delta s,\infty)$ for some $\delta\ge0.
\end{equation*} 
Therefore, for any $f\in \cE$,
\begin{align}\label{local-est4}
\big \lan N_{f,0s} \big\ran_0 \le\big \langle N_{|x|>b}\big\rangle_0 .\end{align}

Next, retaining the first term in \eqref{propag-est2}, dropping the second one and using \eqref{local-est4}, we obtain that
\begin{align*} 
&\big\lan N_{f,ts} \big \ran_t \le \,\big \langle N_{|x|>b}\big\rangle_0 +C_{f,n} \sum_{k=2}^n s^{-k}\int_0^t \big\lan N_{\tilde f'_k,rs}\big \ran_r dr + C_{f,n} t s^{-n-1}\langle N\rangle_0. 
\end{align*}
Applying \eqref{propag-est31} and again \eqref{local-est4} to estimate the integrated term, we deduce that
\begin{align} \label{propag-est4} 
\big\lan N_{f,ts} \big\ran_t \le(1+C_{f,c,n,\delta_0}s^{-1}) \big\langle N_{|x|>b} \big\rangle_0 +  C_{f,c,n,\delta_0} s^{-n}  \langle N\rangle_0,  \end{align}
for any $n\le p -1$ and $\max(t,c^{-1}\delta_0)\le s$.

Now, for any $f\in \cE $, we have $f(\lam)=1$ for $\lam\ge c-v$, and therefore $f(\frac{\cdot-b-vt}{s})=1$ on 
$[ b+vt +(c-v) s , \infty)$. 
For $\rho \ge b + cs$ and $s\ge t$, we have $[\rho,\infty) \subset [ b+vt +(c-v) s , \infty)$.
Hence, choosing $s=(\rho-b)/c\ge\max(t,c^{-1}\delta_0)$,  we conclude that, for $\rho \ge b + ct$,
\begin{align*}
\big\lan N_{|x|>\rho} \big\ran_t &\le \big\lan N_{f,ts}\big \ran_t \\
&\le\big(1+C_{f,c,n,\delta_0}(\rho-b)^{-1}\big) \big\langle N_{|x|>b} \big\rangle_0 +  C_{f,c,n,\delta_0} (\rho-b)^{-n}  \langle N\rangle_0.
\end{align*}
 Since $ \psi_t :=e^{-iHt}\psi_0$,  this implies \eqref{max-vel-est}. 
\end{proof}

\begin{proof}[Proof of \eqref{max-vel-est_2}]
It suffices to proceed in the same way, with the following modifications. Fix $c> \kappa$, $v=\frac12(c+\kappa)$ and  let $\cG$ be the set of functions  $0\le f\in \mathrm{C}^\infty(\R)$, supported in $(-\infty,0)$ and
satisfying $f(\lam)=1$ for $\lam\le v-c$,  and $f^\prime\le 0$, with $\sqrt{|f'|}\in \mathrm{C}^\infty(\mathbb{R})$. In other words, $f\in\cG$ if and only if $f(-\cdot)\in\cE$.

We can then adapt the proof of Proposition \ref{prop:propag-est1} and \eqref{max-vel-est}, considering the time dependent observable 
\begin{align}\label{propag-obs1_2}
\Phi_s(t) & = \mathrm{d}\Gamma\big ( f(|x|^-_{ts}) \big ), \quad f \in \cG, \quad |x|^-_{ts} :=s^{-1}(|x| -\rho +v t),
\end{align}
instead of \eqref{propag-obs1}.
\end{proof}

\subsection{Propagation bounds in annuli}

In order to prove Theorems \ref{thm:MVE-SM} and \ref{thm:MVE-SM2}, we need the following extension of Theorem \ref{thm:max-vel-est}.

\begin{theorem}\label{thm:max-vel-est-annuli}
Suppose that \eqref{eq:kappa_J} is satisfied for some integer $p \ge 1$. For all $c > \kappa$, $\delta_0>0$ and all integers $n\le p-1$, there exists $C>0$ such that, for all $b>0$, $0<\beta<1$, $\rho >\beta^{-1}(b + \delta_0)$ and $\psi_0\in\mathcal{D}(N^{\frac12})\subset\mathcal{F}$,
\begin{align*}
&\sup_{0\le t<c^{-1}(\beta\rho-b)} \big \langle \psi_t, N_{\beta\rho< |x|<(2-\beta)\rho} \psi_t \big \rangle\notag\\
& \quad\le \big(1+C(\beta\rho-b)^{-1}\big) \big\langle\psi_0, N_{b<|x|<2\rho-b}\psi_0\big\rangle +  C (\beta\rho-b)^{-n}  \langle\psi_0,N\psi_0\rangle.
\end{align*}
\end{theorem}

The prove Theorem \ref{thm:max-vel-est-annuli} we will use the time-dependent observable 
\begin{align*}
\Phi_s(t) & = \mathrm{d}\Gamma\big ( f(|x|^+_{ts})g(|x|^-_{ts}) \big ), 
\end{align*}
where $f\in\cE$, $g\in\cG$ and
\begin{align*}
 |x|^+_{ts} :=s^{-1}(|x| -b -v t),\ \quad |x|^-_{ts} :=s^{-1}(|x| - (2\rho-b) +v t).
\end{align*}
The proof of Theorem \ref{thm:max-vel-est-annuli} is then similar to that of Theorem \ref{thm:max-vel-est}. It is deferred to Appendix \ref{app:annuli}.

\section{Proof of Theorem \ref{thm:MVE-SM}}
\label{sec:thm2}

The overarching idea of the proof is to convert, with help of a new localization technique, the particle number bounds obtained in  Theorem \ref{thm:max-vel-est} into bounds on the commutators of observables stated in  Theorem \ref{thm:MVE-SM}. More precisely, we factorize the Fock space as $\mathcal{F}=\mathcal{F}_<\otimes\mathcal{F}_>$, where $\mathcal{F}_<$ is the Fock space over $\ell^2(\{|x|<\rho\})$ and $\mathcal{F}_>$ is the Fock space over $\ell^2(\{|x|\ge\rho\})$. The localized observables $A$ and $B$ factorize in this representation as $A=A\otimes\one$, $B=\one\otimes B$. Next we compare the dynamics generated by the Hamiltonian $H$ to the dynamics generated by an uncoupled Hamiltonian of the form $\tilde H=H_<\otimes\one+\one\otimes H_>$ (precise definitions will be given below). This produces error terms that we can control thanks to the dynamical bounds on particle transport established in Theorem \ref{thm:max-vel-est-annuli}. We then deduce that
\begin{equation*}
A_t B = e^{-itH}Ae^{itH}B\approx \big(e^{-itH_<}Ae^{itH_<}\big)\otimes B,
\end{equation*}
up to small remainder terms. Since the same holds if $A_t B$ in the left-hand-side is replaced by $B A_t$, this finally implies the Lieb-Robinson bound stated in Theorem \ref{thm:MVE-SM}.

Recalling the parameters $c$, $b$, $\delta_0$ and $\rho$ involved in the statement of Theorem \ref{thm:MVE-SM}, we introduce two further parameters that will be fixed as follows throughout our proof. We let $v$ be such that $c>v>\kappa$ and introduce a parameter $\alpha\in(0,1)$, $\alpha$ close to $1$ such that 
\begin{equation}\label{eq:cond_alpha}
c\big (1 + (2\alpha-2)(1+\delta_0^{-1})\big)>v.
\end{equation}
Let 
\begin{equation*}
0\le t<(2c)^{-1}(\rho-b)<(2v)^{-1}\big((2\alpha-1)\rho-b\big).
\end{equation*}
Note that the second inequality above is a consequence of \eqref{eq:cond_alpha} together with the fact that $\rho>(1+\delta_0)b$.

We divide the proof of Theorem \ref{thm:MVE-SM} into a few subsections.

\subsection{Factorization of Fock space}

 Let $H \equiv H_\Lam$ and define the symmetric Fock spaces (cf. \eqref{Fock-sp} below)
\begin{align*}
&\cF_< := \cF( \ell^2( \mathcal{B}(0,\rho) ) ),\\ 
&\cF_{>} := \cF( \ell^2( \mathcal{B}(0,\rho)^c )),
\end{align*}
over $\ell^2( \mathcal{B}(0,\rho) )$ and  $\ell^2( \mathcal{B}(0,\rho)^c )$, respectively. For $f\in \ell^2(\Lambda)$, we write
\begin{equation*}
f_<:=\chi_{|x|<\rho}f, \quad f_>:=\chi_{|x|\ge\rho}f.
\end{equation*}
 Let $U_\rho : \cF\to \cF_< \otimes \cF_>$ be the unitary operator defined by
 \begin{equation*}
 U_\rho\Omega:=\Omega_{<}\otimes\Omega_>,
 \end{equation*}
where $\Omega_\sharp$ is the vacuum in $\cF_\sharp$ and, for all $f\in \ell^2(\Lambda)$,
 \begin{equation}\label{eq:Urhodefn} 
 U_\rho a^\sharp(f)=\big(a^\sharp(f_<)\otimes\one+\one\otimes a^\sharp(f_>) \big)U_\rho,
 \end{equation}
 where $a^\sharp$ stands for $a$ or $a^*$. Note that since $A$ is localized in $\mathcal{B}(0,b)\subset\mathcal{B}(0,\rho)$ and $B$ is localized in $\mathcal{B}(0,2\rho-b)^c\subset\mathcal{B}(0,\rho)^c$, we have
 \begin{align*}
 U_\rho AU_\rho^*=A\otimes\one, \quad U_\rho BU_\rho^*=\one\otimes B.
 \end{align*}
To complete our construction, we define the Hamiltonians 
\begin{align*}
H_<:=H_{\mathcal{B}(0,\rho)}, \quad H_>:=H_{\mathcal{B}(0,\rho)^c},
\end{align*}
where $H_{\Lam'}$, for $\Lam'\subset\Lam$, are defined by \eqref{HLam}, with $\Lam$ replaced by $\Lam'$, and
\begin{align} \label{Htilde}\tilde H:=H_<\otimes\one+\one\otimes  H_> .
\end{align}

\subsection{Approximating the full dynamics by the decoupled dynamics}

In this section we estimate $U_\rho e^{itH}U_\rho^*$ by approximating it by $e^{it\tilde H}$:  We claim that, for any $\psi_1,\psi_2\in U_\rho\mathcal{D}(N^{\frac12}) \subset \cF_<\otimes\cF_>$,
\begin{align}
&\big|\langle\psi_1, U_\rho e^{itH}U_\rho^*\psi_2\rangle - \big\langle\psi_1,e^{it\tilde H}\psi_2\big\rangle\big |\notag\\
&\le C t \sup_{0\le r\le t}  \big( (1-\alpha)\rho)^{-\frac{p}{2}}\| N^{\frac12} U_\rho^*\psi_1\|+ \big \| (N_{\alpha,\rho}^{(1)})^{\frac12} U_\rho^*e^{-i(t-r)\tilde H} \psi_1 \big \| \big) \notag \\
&\quad\quad\times\big((1-\alpha)\rho)^{-\frac{p}{2}}\| N^{\frac12}e^{irH}U_\rho^*\psi_2\|+\big\| (N_{\alpha,\rho}^{(1)})^{\frac12} e^{irH}U_\rho^*\psi_2\big\| \big) , \label{eq:commut3}
\end{align}
where, to shorten notations, we have set
\begin{equation*}
N_{\alpha,\rho}^{(1)} := N_{\mathcal{C}_{\alpha\rho,(2-\alpha)\rho}}
\end{equation*}
We recall that $\mathcal{C}_{b_1,b_2} = \{ x \in \Lambda , b_1\le|x|\le b_2 \}$.

In order to prove this claim, we use the fundamental theorem of calculus to compute
\begin{align}
&\langle\psi_1,U_\rho e^{itH}U_\rho^*\psi_2\rangle - \big\langle\psi_1,e^{it\tilde H}\psi_2\big\rangle\notag\\
&=-i\int_0^t \big\langle \psi_1 , e^{i(t-r)\tilde H} \{ \tilde H U_\rho - U_\rho H \} e^{irH}U_\rho^*\psi_2\big\rangle dr. \label{eq:commut1}
\end{align}
A direct computation gives 
\begin{align}
& \tilde HU_\rho-U_\rho H  = U_\rho \Big ( \sum_{\substack{|x|\ge\rho,|y|<\rho}} J_{xy} b_x^*b_y +  \mathrm{h.c.} \Big ). \label{eq:commut2}
\end{align}
Using the Cauchy-Schwarz inequality, we have, for all $\varphi_1,\varphi_2\in\mathcal{D}(N^{1/2})\subset\mathcal{F}$,
\begin{align}
&\Big|\Big \langle \varphi_1 , \sum_{\substack{|x|\ge\rho,|y|<\rho}} J_{xy} b_x^*b_y \varphi_2 \Big \rangle \Big | \notag \\
&\le \Big(\sum_{\substack{|x|\ge\rho,|y|<\rho}} |J_{xy}| \langle \varphi_1,n_x\varphi_1\rangle\Big)^{\frac12}\Big(\sum_{|x|\ge\rho,|y|<\rho}|J_{xy}|\langle\varphi_2,n_y\varphi_2\rangle\Big)^{\frac12} . \label{eq:commut21}
\end{align}
For the second term in the right-hand-side of \eqref{eq:commut21}, we write
\begin{align}
&\sum_{\substack{|x|\ge\rho,|y|<\rho}}|J_{xy}|\langle\varphi_2,n_y\varphi_2\rangle \notag \\ 
&=\sum_{\substack{|y|<\alpha\rho}} \langle\varphi_2,n_y\varphi_2\rangle \sum_{\substack{|x|\ge\rho}} |J_{xy}| + \sum_{\substack{\alpha\rho\le|y|<\rho}} \langle\varphi_2,n_y\varphi_2\rangle \sum_{\substack{|x|\ge\rho}} |J_{xy}| \notag \\ 
&\le ((1-\alpha)\rho)^{-p} \sum_{\substack{|y|<\alpha\rho}} \langle\varphi_2,n_y\varphi_2\rangle \sum_{\substack{|x|\ge\rho}} |J_{xy}||x-y|^{p } \notag \\
&\quad+ \sum_{\substack{\alpha\rho\le|y|<\rho}} \langle\varphi_2,n_y\varphi_2\rangle \sum_{\substack{|x|\ge\rho}} |J_{xy}| \notag \\ 
&\le \kappa_J^{(p )}((1-\alpha)\rho)^{-p } \langle\varphi_2,N\varphi_2\rangle + \kappa_J^{(0)} \langle\varphi_2,N_{\mathcal{C}_{\alpha\rho,\rho}}\varphi_2\rangle . \label{eq:commut23_0}
\end{align}
Similarly, the first term in the right-hand-side of \eqref{eq:commut21} can be estimated as follows:
\begin{align}
&\sum_{\substack{|x|\ge\rho,|y|<\rho}}|J_{xy}|\langle\varphi_1,n_x\varphi_1\rangle \notag \\ 
&=\sum_{\substack{|x|>(2-\alpha)\rho}} \langle\varphi_1,n_x\varphi_1\rangle \sum_{\substack{|y|<\rho}} |J_{xy}| + \sum_{\substack{\rho\le|x|\le(2-\alpha)\rho}} \langle\varphi_1,n_x\varphi_1\rangle \sum_{\substack{|y|<\rho}} |J_{xy}| \notag \\ 
&\le ((1-\alpha)\rho)^{-p} \sum_{\substack{|x|>(2-\alpha)\rho}} \langle\varphi_1,n_x\varphi_1\rangle \sum_{\substack{|y|<\rho}} |J_{xy}||x-y|^{p } \notag \\
&\quad+ \sum_{\substack{\rho\le|x|\le(2-\alpha)\rho}} \langle\varphi_1,n_x\varphi_1\rangle \sum_{\substack{|y|<\rho}} |J_{xy}| \notag \\ 
&\le \kappa_J^{(p )}((1-\alpha)\rho)^{-p } \langle\varphi_1,N\varphi_1\rangle + \kappa_J^{(0)} \langle\varphi_1,N_{\mathcal{C}_{\rho,(2-\alpha)\rho}}\varphi_1\rangle . \label{eq:commut23}
\end{align}
Inserting \eqref{eq:commut21}--\eqref{eq:commut23} into \eqref{eq:commut2}, we obtain, for any $\varphi_1,\varphi_2\in\mathcal{D}(N^{1/2})\subset\mathcal{F}$,
\begin{align}
&\big|\big\langle U_\rho\varphi_1, (\tilde HU_\rho-U_\rho H)\varphi_2\big\rangle\big|\notag\\
&\le C\big((1-\alpha)\rho)^{-\frac{p}{2}}\| N^{\frac12}\varphi_1\|+\| (N_{\alpha,\rho}^{(1)})^{\frac12}\varphi_1\|\big)\notag\\
&\quad\times\big((1-\alpha)\rho)^{-\frac{p}{2}}\| N^{\frac12}\varphi_2\|+\|(N_{\alpha,\rho}^{(1)})^{\frac12}\varphi_2\|\big). \label{eq:commut24}
\end{align}
Using \eqref{eq:commut1}, \eqref{eq:commut2}, \eqref{eq:commut24} and the fact that $N$ commutes with $H$, we obtain \eqref{eq:commut3}.

\subsection{Approximating $A_tB$ by $(e^{-itH_<}Ae^{itH_<})\otimes B$ and conclusion of the proof}
We proceed in two steps.

\medskip 
\noindent \textbf{Step 1.} We approximate the first evolution operator $e^{-itH}$ by $e^{-it\tilde H}$ in $A_tB$: Let $A$, $B$ and $\varphi\in\mathcal{D}(N^{\frac12})\subset\cF$ be as in the statement of the theorem. We claim that
\begin{align}
\big\langle\varphi,A_tB\varphi\big\rangle&=\big\langle U_\rho\varphi,e^{-it\tilde H}(A\otimes\one)U_\rho e^{itH}U_\rho^*(\one\otimes B)U_\rho\varphi\big\rangle\notag\\
&\quad+\langle\varphi,\mathrm{Rem}_1\varphi\rangle, \label{eq:new1}
\end{align}
with 
 \begin{align}
&\big|\big\langle\varphi, \mathrm{Rem}_1 \,\varphi\big\rangle\big| \le C_{c,p,\delta_0} \|A\| \|B\| t (\rho-b)^{-p+1} \big\langle \varphi , N \varphi \big\rangle . \label{eq:atAB_33_0}
\end{align}

In order to prove \eqref{eq:new1}--\eqref{eq:atAB_33_0}, we write
\begin{align}
&A_tB=e^{-itH}Ae^{itH}B = U_\rho^*U_\rho e^{-itH}U_\rho^*(A\otimes\one)U_\rho e^{itH}U_\rho^*(\one\otimes B)U_\rho.\label{eq:atAB}
\end{align}
Applying \eqref{eq:commut3} with $\psi_1=U_\rho\varphi$, $\psi_2=(A\otimes\one)U_\rho e^{itH}U_\rho^*(\one\otimes B)U_\rho\varphi$ gives \eqref{eq:new1} with
\begin{align}
&\big|\big\langle\varphi, \mathrm{Rem}_1 \,\varphi\big\rangle\big|  \notag\\
&\le C t \sup_{0\le r\le t} \big((1-\alpha)\rho)^{-\frac{p}{2}}\|N^{\frac12}\varphi\|+\big\| (N_{\alpha,\rho}^{(1)} )^{\frac12} U_\rho^*e^{i(t-r)\tilde H} U_\rho\varphi \big \|\big) \notag\\
&\times \big((1-\alpha)\rho)^{-\frac{p}{2}}\big\|N^{\frac12}e^{-irH}A e^{itH}B\varphi\big\|+\big\| (N_{\alpha,\rho}^{(1)} )^{\frac12} e^{-irH} A e^{itH}B\varphi\big \| \big)  . \label{eq:atAB_31}
\end{align}

Since $N$ commutes with $H$, $A$ and $B$, the third term in \eqref{eq:atAB_31} can be estimated as
\begin{align}
\big\|N^{\frac12}e^{-irH}A e^{itH}B\varphi\big\|&=\big\|N^{\frac12}A e^{itH}B\varphi\big\|\le\|A\| \|B\| \big\| N^{\frac12}\varphi\big\|. \label{eq:P1_0}
\end{align}

To estimate the last term in \eqref{eq:atAB_31}, we use Theorem \ref{thm:max-vel-est-annuli}. Since $0\le r\le t<\frac12v^{-1}(\rho(2\alpha-1)-b)$, we have $\alpha\rho>\frac{b}{2}+\frac{\rho}{2}+vr$. Hence we can apply Theorem \ref{thm:max-vel-est-annuli} (with $b'=\frac{b}{2}+\frac{\rho}{2}$, noticing then that $\alpha\rho-b'=(\alpha-\frac12)\rho-\frac{b}{2}\ge C_{\delta_0}(\rho-b)$ by \eqref{eq:cond_alpha}), which yields
\begin{align} 
&\big\| (N_{\alpha,\rho}^{(1)})^{\frac12} e^{-irH} A e^{itH}B\varphi\big \|\notag\\
& \le C_{c,p,\delta_0} \Big( \big\| (N_{\alpha,\rho}^{(2)} )^{\frac12} A e^{itH}B\varphi\big\| +(\rho-b)^{-\frac{p-1}{2}}\big\| N^{\frac12} A e^{itH}B\varphi\big\|\Big) \notag\\
& \le C_{c,p,\delta_0} \Big( \big\| (N_{\alpha,\rho}^{(2)} )^{\frac12} A e^{itH}B\varphi\big\| +(\rho-b)^{-\frac{p-1}{2}}\|A\|\|B\| \big\| N^{\frac12} \varphi\big\|\Big) , \label{eq:l2_1}
\end{align}
where we used  \eqref{eq:P1_0} in the last equation and set
\begin{equation*}
N_{\alpha,\rho}^{(2)}:=N_{\mathcal{C}_{\frac{b}{2}+\frac{\rho}{2},\frac{3\rho}{2}-\frac{b}{2}}} .
\end{equation*}

To control the first term in the right-hand-side of \eqref{eq:l2_1}, we first observe that $A$ commutes with $N_{\alpha,\rho}^{(2)}$ since $A$ is localized in $\mathcal{B}(0,b)$ and $b<b/2+\rho/2$. Hence
\begin{align}\label{eq:comm_A_N}
&\big\| (N_{\alpha,\rho}^{(2)} )^{\frac12} A e^{itH}B\varphi\big\| =\|A\| \big\| (N_{\alpha,\rho}^{(2)})^{\frac12} e^{itH} B\varphi\big\|.
\end{align}
 Next we apply again Theorem \ref{thm:max-vel-est-annuli}. We fix $0<\beta<1$ such that $\beta\rho=b/2+\rho/2$ and note that $\beta\rho>b+(1-\alpha)\rho+vt$ (since $0\le t<\frac12v^{-1}(\rho(2\alpha-1)-b)$). Therefore, Theorem \ref{thm:max-vel-est-annuli} (with $b'=(1-\alpha)\rho+b$) implies
\begin{align}
&\big\| (N_{\alpha,\rho}^{(2)} )^{\frac12} e^{itH}B\varphi\big\| \le C_{c,p,\delta_0} \Big(\big\| (N_{\alpha,\rho}^{(3)} )^{\frac12} B\varphi\big\| + (\rho-b)^{-\frac{p-1}{2}}\big\| N^{\frac12} B\varphi\big\|\Big) , \label{eq:l2_2}
\end{align}
where
\begin{equation*}
N_{\alpha,\rho}^{(3)}:=N_{\mathcal{C}_{b+(1-\alpha)\rho,(1+\alpha)\rho-b}}.
\end{equation*}
Since $B$ is localized in $\mathcal{B}(0,2\rho-b)^c\subset \mathcal{B}(0,(1+\alpha)\rho-b)^c$, it commutes with $N_{\alpha,\rho}^{(3)}$. Moreover, using the condition \eqref{eq:hyp:empty} and the fact that $\alpha<1$, we have $N_{\alpha,\rho}^{(3)}\varphi=0$. Hence, using in addition that $B$ commutes with $N$, \eqref{eq:l2_2} reduces to
\begin{align}
&\big\| (N_{\alpha,\rho}^{(2)} )^{\frac12} e^{itH}B\varphi\big\| \le C_{c,p,\delta_0} (\rho-b)^{-\frac{p-1}{2}} \|B\| \big\| N^{\frac12} \varphi\big\| . \label{eq:l2_2_0}
\end{align}

Putting together Eq. \eqref{eq:l2_1}--\eqref{eq:l2_2_0}, we have proven that
\begin{align}
&\big\| (N_{\alpha,\rho}^{(1)} )^{\frac12} e^{-irH} A e^{itH}B\varphi\big \| \le C_{c,p,\delta_0} (\rho-b)^{-\frac{p-1}{2}}\|A\|\|B\| \big\| N^{\frac12} \varphi\big\| . \label{eq:estim_2_0}
\end{align}

The term $\| (N_{\alpha,\rho}^{(1)} )^{\frac12} U_\rho^*e^{i(t-r)\tilde H} U_\rho\varphi\|$ in \eqref{eq:atAB_31} can be treated in the same way (the argument is actually simpler since this term contains only one propagator). This gives
\begin{align}
&\big\| (N_{\alpha,\rho}^{(1)} )^{\frac12} U_\rho^*e^{i(t-r)\tilde H} U_\rho\varphi \big \| \le C_{c,p,\delta_0}\big((1-\alpha)\rho+b\big)^{-\frac{p-1}{2}}\big\|N^{\frac12}\varphi\big\|. \label{eq:estim_3_0}
\end{align}
 
 Inserting \eqref{eq:P1_0}, \eqref{eq:estim_2_0} and \eqref{eq:estim_3_0} into \eqref{eq:atAB_31}, using in addition that $(1-\alpha)\rho+b\ge(1-\alpha)(\rho-b)$ and that $\alpha$ can be fixed such that $(1-\alpha)^{-1}\le C_{c,\delta_0}$ (see \eqref{eq:cond_alpha}), we obtain \eqref{eq:atAB_33_0}.

\medskip

\noindent \textbf{Step 2.} Now we approximate the second evolution operator, $e^{itH}$, by $e^{it\tilde H}$ in $A_tB$. 

Going back to \eqref{eq:new1}, we apply again \eqref{eq:commut3}, now with $\psi_1=(A\otimes\one)e^{it\tilde H}U_\rho\varphi$, $\psi_2=(\one\otimes B)U_\rho\varphi$. This gives
\begin{align}
\big\langle\varphi,e^{-itH}Ae^{itH}B\varphi\big\rangle&=\big\langle U_\rho\varphi,e^{-it\tilde H}(A\otimes\one) e^{it\tilde H}(\one\otimes B)U_\rho\varphi\big\rangle\notag\\
&\quad+\langle\varphi,\mathrm{Rem}_1\varphi\rangle+\langle\varphi,\mathrm{Rem}_2\varphi\rangle, \label{eq:new1_1}
\end{align}
with 
\begin{align}
&\big|\big\langle\varphi, \mathrm{Rem}_2 \,\varphi\big\rangle\big|  \le C t \sup_{0\le r\le t} \big((1-\alpha)\rho)^{-\frac{p}{2}}\|N^{\frac12}AU_\rho^*e^{it\tilde H}U_\rho\varphi\| \notag \\
&\qquad+\big\| (N_{\alpha,\rho}^{(1)} )^{\frac12} U_\rho^*e^{i(t-r)\tilde H} (A\otimes\one)e^{it\tilde H}U_\rho\varphi \big \|\big) \notag\\
&\qquad\times \big((1-\alpha)\rho)^{-\frac{p}{2}}\big\|N^{\frac12}e^{irH}B\varphi\big\|+\big\| (N_{\alpha,\rho}^{(1)} )^{\frac12} e^{irH}B\varphi\big \| \big)  . \label{eq:atAB_31_1}
\end{align}
Proceeding in the same way as for $\mathrm{Rem}_1$ gives
\begin{align}
&\big|\big\langle\varphi, \mathrm{Rem}_2 \,\varphi\big\rangle\big| \le C_{c,p,\delta_0} \|A\| \|B\| t (\rho-b)^{-p+1} \big\langle \varphi , N \varphi \big\rangle. \label{eq:atAB_34_0}
\end{align}

Inserting \eqref{eq:atAB_33_0} and \eqref{eq:atAB_34_0}  into \eqref{eq:new1_1} and using that, due to \eqref{Htilde}, $e^{-it\tilde H}(A\otimes\one)e^{it\tilde H}=( e^{-itH_<}Ae^{itH_<} ) \otimes \one$, we can conclude that 
 \begin{align}
e^{-itH}Ae^{itH}B &=  U_\rho^*e^{-it\tilde H}(A\otimes\one)e^{it\tilde H}(\one\otimes B)U_\rho  + \mathrm{Rem} \notag\\ &=  U_\rho^*( e^{-itH_<}Ae^{itH_<}  \otimes B)U_\rho  + \mathrm{Rem} , \label{eq:atAB_5}
\end{align}
with
\begin{align}
&\big|\big\langle\varphi, \mathrm{Rem}\,\varphi\big\rangle\big|\le C_{c,p,\delta_0} t (\rho-b)^{1-p} \|A\| \|B\| \big\langle \varphi , N \varphi \big\rangle. \label{eq:rem_est_0}
\end{align}

\medskip

\noindent \textbf{Conclusion of the proof.}
Proceeding in the same way, one shows that
\begin{align}
& Be^{-itH}Ae^{itH} = U_\rho^*(e^{-itH_<}Ae^{itH_<}\otimes B)U_\rho  + \mathrm{Rem}  \label{eq:atAB_6}
\end{align}
with $\mathrm{Rem}$ satisfying \eqref{eq:rem_est_0}. Equations \eqref{eq:atAB_5}, \eqref{eq:rem_est_0} and \eqref{eq:atAB_6} prove the theorem.
\qed

\begin{proof}[Proof of Theorem \ref{thm:localapprox}]
Taking $B=\one$, Theorem \ref{thm:localapprox} directly follows from \eqref{eq:atAB_5} and \eqref{eq:rem_est_0}, with
\begin{equation}
[A_t]_\rho :=U_\rho^*e^{-it\tilde H}(A\otimes\one)e^{it\tilde H}U_\rho \equiv U_\rho^*( e^{-itH_<}Ae^{itH_<}  \otimes \one)U_\rho .\label{eq:def_alpha_trho}
\end{equation}
\end{proof}

\appendix

\section{Self-adjointness and basic commutators}\label{sec:SA}
In this appendix, we prove the self-adjointness of the Bose-Hubbard Hamiltonian, which is non-trivial in the case of an infinite lattice $\Lambda$.

We define formally the bosonic Fock space over $\ell^2(\mathcal{L})$ as
\begin{equation}\label{Fock-sp}
\mathcal{F}:=\bigoplus_{n=0}^\infty \mathcal{F}_n \quad \text{with} \quad \mathcal{F}_n := \otimes_s^n \ell^2(\mathcal{L}) ,
\end{equation}
where $\otimes_s$ denotes the symmetric tensor product. 
 Note that
\begin{equation*}
\mathcal{F}_n = \mathrm{Ran}(\chi_{N=n}) \text{ is the $n$-particle space.}
\end{equation*}
We write $\varphi = (\varphi_n)$ for each $\varphi \in \mathcal{F}$, with $\varphi_n\in\mathcal{F}_n$. For convenience, we will use the notations
\begin{equation*}
 N_2 := \sum_{x\in\Lam} n_x^2 , \quad T_J := \sum_{x\in\Lambda,y\in\Lambda} J_{xy} b_x^*b_y,
\end{equation*}
so that
\begin{align*}
H  = \frac{g}{2} N_2 - \Big ( \frac{g}{2} + \mu \Big ) N - T_J .
\end{align*} 

We first establish self-adjointness.

\begin{proposition}\label{prop:sa}
The operator $H$ is self-adjoint on $\mathcal{F}$ with domain 
\begin{equation*}
\mathcal{D}(H)=\Big \{\varphi = (\varphi_{n}) \in \mathcal{F} \, \text{ such that } \,  \sum_{n=0}^\infty \big\|H \varphi_n \big\|^2 < \infty \Big \}.
\end{equation*}
\end{proposition}

\begin{proof}
Observe that $H$ preserves the number of particles. For all $n\in\mathbb{N}_0$, let $H_{n}$ be the restriction of $H$ to $\mathcal{F}_n$. We have
\begin{align*}
H_{n} =\frac{g}{2} N_{2,n} - \Big ( \frac{g}{2} + \mu \Big ) n + T_{J,n},
\end{align*} 
where $N_{2,n}$, $T_{J,n}$ are the restrictions of $N_2$, $T_J$ to $\mathcal{F}_n$. Moreover, by the Cauchy-Schwarz inequality, for all $\varphi\in\mathcal{F}_n$,
\begin{align*}
\big | \langle \varphi , T_{J,n} \varphi \rangle | &\le \sum_{x,y\in\Lambda} \big |J_{xy}|  |\langle \varphi , b^*_x b_y \varphi \rangle \big | \\
& \le \Big( \sum_{x,y\in\Lambda} |J_{xy}| \big \| b_x \varphi \big \|^2 \Big)^{\frac12} \Big( \sum_{x,y\in\Lambda}|J_{xy}|\big \| b_y \varphi \big \|^2 \Big)^{\frac12} \\
& \le \kappa_J^{(0)} \langle \varphi , N \varphi \rangle = n\kappa_J^{(0)} \|\varphi \|^2.
\end{align*}
Hence $T_{J,n}$ is bounded. Likewise,
\begin{align*}
\big | \langle \varphi , N_{2,n} \varphi \rangle | = \sum_{x\in\Lambda}\langle \varphi , n_x^2 \varphi \rangle \le \sum_{x,y\in\Lambda}\langle \varphi , n_xn_y \varphi \rangle = \langle \varphi , N^2 \varphi \rangle =n^2\|\varphi\|^2.
\end{align*}
This shows that $H_{n}$ is a bounded operator on $\mathcal{F}_n$. Since in addition $H_n$ is clearly symmetric, it is then self-adjoint. The proposition follows.
\end{proof}

Next we note the following 

\begin{lemma}\label{lem:HRf-com} Let $f:\R\to\R$ be a bounded mesurable function. In the sense of forms on $\mathcal{D}(H)\cap\mathcal{D}(N)$, we have  
\begin{align} \label{HRf-com}[H, \mathrm{d}\Gamma(f(x))]&=\sum_{x\in\Lambda,y\in\Lambda}J_{x,y}\{f(x)-f(y) \}b_x^*b_y.\end{align}\end{lemma}

\begin{proof} 
Note that $\mathrm{d}\Gamma(f(x))$ commutes with $N$ and $N_2$. Hence,
\begin{align} \label{H0Rf-com}
[H,\mathrm{d}\Gamma(f(x))] = [T, \mathrm{d}\Gamma(f(x))]&=-\sum_{x\in\Lambda,y\in\Lambda}\sum_{z\in\Lambda}J_{xy} [b_x^*b_y, f(z) b_z^*b_z],\end{align}
which due to the commutation relations between $b_z$ and $b_z^*$ gives \eqref{HRf-com}.
\end{proof}
 Taking $f(x)\equiv 1$ in \eqref{HRf-com}, we arrive at 
\begin{corollary}\label{cor:HN-com} $H$ commutes with $N$.\end{corollary}

\section{Proof of Lemma \ref{lm:taylor}}\label{app:lemma}

Recall that $c> v>\kappa$ and that $\cE$ denotes the set of functions  $0\le f\in \mathrm{C}^\infty(\R)$, supported in $\R^+=(0,\infty)$ and
satisfying $f(\lam)=1$ for $\lam\ge c-v$,  and $f^\prime\ge 0$, with $\sqrt{f'}\in \mathrm{C}^\infty(\R)$. We say a function $h$ is {\it admissible} if it is smooth, non-negative with $\supp h\subset (0, c-v)$ and $\sqrt{h}\in \mathrm{C}^\infty(\R)$. 
Note that if $h$ is admissible, then $h \le C f'$, with $f\in\cE$. Indeed, writing
\[f_1(\lam)=\int_{-\infty}^\lam h(s)ds,\ \text{ we have }\ f_1/\int_{-\infty}^\infty h(s)ds\in \cE.\]  
Similarly, if $f_1,f_2\in\cE$, then $f_1+f_2\lesssim f_3$ for some $f_3\in\cE$.
If $g,g_1\in\mathrm{C}_0^\infty(\R)$, we write $g\prec g_1$ if $g_1=1$ on $\mathrm{supp}(g)$.

\begin{proof}[Proof of Lemma \ref{lm:taylor}]
We use a Taylor expansion
\begin{align*}
f(x)-f(y)=\sum_{k=1}^n \frac{ (x-y)^k }{ k! } f^{(k)}(x) + \mathcal{O}( (x-y)^{n+1} ).
\end{align*}
The term for $k=1$ is rewritten as
\begin{align*}
(x-y)f'(x) &= (x-y)u(x)u(y) + (x-y)u(x)(u(x)-u(y)).
\end{align*}
If $n=1$, the lemma follows. Now, assuming $n\ge2$, we use again a Taylor expansion for the second term in the previous equation, which yields
\begin{align*}
& (x-y)u(x)(u(x)-u(y)) \\
&= \sum_{\ell=1}^{n-1} \frac{ (x-y)^{\ell+1} }{ \ell! } u(x)u^{(\ell)}(x) + \mathcal{O}( (x-y)^{n+1} ).
\end{align*}
Combining the previous equations gives
\begin{align*}
f(x)-f(y)&=(x-y)u(x)u(y)+\sum_{k=2}^n (x-y)^k \Big ( \frac{  f^{(k)}(x) }{ k! } + \frac{ u(x)u^{(k-1)}(x) }{ (k-1)! } \Big ) \\
&\quad+ \mathcal{O}( (x-y)^{n+1} ).
\end{align*}
Consider the term for $k=2$. Let $v_2\in\mathrm{C}_0^\infty(\mathbb{R})$ be such that $\mathrm{supp}(v_2)\subset(0,c-v)$ and $f'\prec v_2$. We write
\begin{align*}
&(x-y)^2 \Big ( \frac{  f^{(2)}(x) }{ 2 } +  u(x)u'(x) \Big ) \\
&=(x-y)^2 v_2(x) \Big ( \frac{  f^{(2)}(x) }{ 2 } +  u(x)u'(x) \Big ) v_2(y) \\
&\quad+ (x-y)^2 v_2(x) \Big ( \frac{  f^{(2)}(x) }{ 2 } +  u(x)u'(x) \Big ) (v_2(x)-v_2(y)) \\
&=(x-y)^2 v_2(x) \Big ( \frac{  f^{(2)}(x) }{ 2 } +  u(x)u'(x) \Big ) v_2(y) \\
&\quad+ v_2(x) \Big ( \frac{  f^{(2)}(x) }{ 2 } +  u(x)u'(x) \Big ) \sum_{\ell=1}^{n-2} \frac{ (x-y)^{\ell+2} }{ \ell! } v_2^{(\ell)}(x) + \mathcal{O}((x-y)^{n+1}).
\end{align*}
Let $w_2(x)=2^{-1}f^{(2)}(x)+u(x)u'(x)$. We have shown that
\begin{align*}
f(x)-f(y)&=(x-y)u(x)u(y)+(x-y)^2v_2(x)w_2(x)v_2(y) \\
&\quad+\sum_{k=3}^n (x-y)^k \Big ( \frac{  f^{(k)}(x) }{ k! } + \frac{ u(x)u^{(k-1)}(x) }{ (k-1)! }\\
&\quad\quad+ v_2(x) \Big ( \frac{  f^{(2)}(x) }{ k! } +  u(x)u'(x) \Big) \frac{ v_2^{(k-2)}(x)}{(k-2)!} \Big ) + \mathcal{O}( (x-y)^{n+1} ).
\end{align*}
Repeating the same procedure iteratively for the terms corresponding to $k \in \{3,\dots,n\}$ in the sum above, we see that there are smooth non-negative functions $v_k$, $2\le k\le n$, such that $\mathrm{supp}(v_k)\subset(0,c-v)$, $f'\prec v_k$ and
\begin{align*}
f(x)-f(y)&=(x-y)u(x)u(y)+\sum_{k=2}^n (x-y)^k v_k(x) w_k(x) v_k(y) \\
&\quad+ \mathcal{O}( (x-y)^{n+1} ) ,
\end{align*}
where $w_k$ are smooth functions (depending on $f^{(\ell)}$, $1\le\ell\le k$) such that $\mathrm{supp}(w_k)\subset(0,c-v)$. Let
\begin{align*}
h_k(x,y):= v_k(x) w_k(x) v_k(y).
\end{align*}
Then $|h_k(x,y)|\lesssim v_k(x) v_k(y)\lesssim (\tilde f_k')^{1/2}(x) (\tilde f_k')^{1/2}(y)$ for some $\tilde f_k\in\cE$ (since $v_k^2$ are admissible functions). Hence the lemma is proven.
\end{proof}

\section{Proof of Theorem \ref{thm:max-vel-est-annuli}}\label{app:annuli}

As in the proof of Theorem \ref{thm:max-vel-est}, we fix $c> \kappa$ and $v=(c+\kappa)/2$. Recall that $\cE$ denotes the set of functions  $0\le f\in \mathrm{C}^\infty(\R)$, supported in $\R^+=(0,\infty)$ and
satisfying $f(\lam)=1$ for $\lam\ge c-v$,  and $f^\prime\ge 0$, with $\sqrt{f'}\in \mathrm{C}^\infty(\mathbb{R})$, while $\cG$ denotes the set of functions $g\in\mathrm{C}^\infty(\R)$ of the form $g=f(-\cdot)$ with $f\in\cE$.

For $f\in\cE$ and $g\in\cG$, we will consider the time-dependent observable 
\begin{align}\label{j1-annuli}
\Phi_s(t) & = \mathrm{d}\Gamma\big ( f(|x|^+_{ts})g(|x|^-_{ts}) \big ), 
\end{align}
where
\begin{align*}
 |x|^+_{ts} :=s^{-1}(|x| -b -v t),\ \quad |x|^-_{ts} :=s^{-1}(|x| - (2\rho-b) +v t).
  \end{align*}

The following proposition corresponds to Proposition \ref{prop:propag-est1} used in the proof of Theorem \ref{thm:max-vel-est}.

\begin{proposition}\label{prop:propag-est1-annuli} 
Suppose that \eqref{eq:kappa_J} is satisfied for some integer $p \ge 1$. For all $c>\kappa$, $f\in \cE$, $g\in\cG$ and any integer $n\le p -1$, there are $j_k\in \cE$, $\ell_k\in\cG$, $2\le k\le n$ and $C>0$ such that, for all $b,t,s>0$ satisfying $0\le t\le s\le c^{-1}(\rho-b)$, we have
\begin{align}
&\int_0^t \big \lan \mathrm{d}\Gamma\big ( f'(|x|^+_{rs} ) - g'(|x|^-_{rs}) \big )\big \ran_r dr\, 
 \le C \Big( 
  s \big\lan \mathrm{d}\Gamma\big ( f(|x|^+_{0s}) g(|x|^-_{0s}) \big ) \big \ran_0 \notag \\
   &\quad+ \sum_{k=2}^n s^{-k+2} \big \lan \mathrm{d}\Gamma\big ( j_k(|x|^+_{0s}) \ell_k(|x|^-_{0s})  \big ) \big \ran_0 +  t s^{-n}\langle N\rangle_0\Big), \label{propag-est31-annuli} 
\end{align}
where the sum should be dropped if $n=0,1$.
\end{proposition} 

\begin{proof}[Proof of Proposition \ref{prop:propag-est1-annuli}]
The structure of the proof is analogous to that of Proposition \ref{prop:propag-est1}. For $n=0$, the result is obvious. Let $n\in\mathbb{N}$ and consider the time-dependent observable \eqref{j1-annuli}. We use \eqref{eq-basic}
and compute the Heisenberg derivative $D\Phi_s(t)$. We have
\begin{align} \label{eq:deriv-annuli}
{\partial\over{\partial t}}\Phi_s(t) = -s^{-1}v \, \big( \mathrm{d}\Gamma\big ( f^\prime(|x|^+_{ts})g(|x|^-_{ts}) - f(|x|^+_{ts})g'(|x|^-_{ts}) \big).
\end{align}
We observe that $f^\prime(|x|^+_{ts})=0$ if $|x|\ge b+vt+(c-v)s$, while $g(|x|^-_{ts})=1$ if $|x|\le2\rho-b-vt+(v-c)s$. Since $t\le s\le c^{-1}(\rho-b)$, we have $b+vt+(c-v)s\le 2\rho-b-vt+(v-c)s$ and therefore
\begin{equation}\label{eq:annuli1}
f^\prime(|x|^+_{ts})g(|x|^-_{ts})=f^\prime(|x|^+_{ts}).
\end{equation}
Likewise,
\begin{equation}
f(|x|^+_{ts})g^\prime(|x|^-_{ts})=g^\prime(|x|^-_{ts}),
\end{equation}
and hence
\begin{align} \label{eq:deriv-annuli2}
{\partial\over{\partial t}}\Phi_s(t) = -s^{-1}v \,  \mathrm{d}\Gamma\big ( f^\prime(|x|^+_{ts}) - g'(|x|^-_{ts}) \big).
\end{align}

Next we compute, using Lemma \ref{lem:HRf-com}, 
\begin{align}
i\big [H , \Phi_s(t) \big ] &= \sum_{x,y\in\Lambda, x\neq y} J_{xy} \{f(|x|^+_{ts})g(|x|^-_{ts})-f(|y|^+_{ts})g(|y|^-_{ts}) \}b_x^*b_y . \label{eq:commut-annuli}
\end{align}
We write
\begin{align}
& f(|x|^+_{ts})g(|x|^-_{ts})-f(|y|^+_{ts})g(|y|^-_{ts}) \notag \\
&=\big(f(|x|^+_{ts})-f(|y|^+_{ts})\big)g(|x|^-_{ts})+f(|y|^+_{ts})\big(g(|x|^-_{ts})-g(|y|^-_{ts})\big) . \label{eq:commut-annuli2}
\end{align}
Proceeding as in the proof of Proposition \ref{prop:propag-est1}, using in particular Lemma \ref{lm:taylor} and the Cauchy-Schwarz inequality, we obtain
\begin{align}
&\Big | \Big \langle \sum_{x,y\in\Lambda, x\neq y}  J_{xy} \big(f(|x|^+_{ts})-f(|y|^+_{ts})\big)g(|x|^-_{ts})b_x^*b_y \Big\rangle_t \Big | \notag \\
&\le  \kappa s^{-1} \big\langle \mathrm{d}\Gamma\big ( f^\prime(|x|^+_{ts}) \big) \big\rangle_t 
+    \sum_{k=2}^n \kappa^{(k)}_J C_{f,k} s^{-k} \big\langle \mathrm{d}\Gamma\big ( \tilde f_k^\prime(|x|^+_{ts}) \big) \big\rangle_t  \notag \\
&\quad+ \kappa_J^{(n+1)} C_{f,n} s^{-n-1}\langle N\rangle_0, \label{eq:annuli2}
\end{align}
for some functions $\tilde f_k\in\cE$. Note that to obtain \eqref{eq:annuli2} we used again \eqref{eq:annuli1} and that, likewise, $\tilde f^\prime_k(|x|^+_{ts})g(|x|^-_{ts})= \tilde f^\prime_k(|x|^+_{ts}).$ In the same way, we have that
\begin{align}
&\Big | \Big \langle \sum_{x,y\in\Lambda, x\neq y}  J_{xy} f(|y|^+_{ts})\big(g(|x|^-_{ts})-g(|y|^-_{ts})\big)b_x^*b_y \Big\rangle_t \Big | \notag \\
&\le  -\kappa s^{-1} \big\langle \mathrm{d}\Gamma\big ( g^\prime(|x|^-_{ts}) \big) \big\rangle_t 
-    \sum_{k=2}^n \kappa^{(k)}_J C_{g,k} s^{-k} \big\langle \mathrm{d}\Gamma\big ( \tilde g_k^\prime(|x|^-_{ts}) \big) \big\rangle_t  \notag \\
&\quad+ \kappa_J^{(n+1)} C_{g,n} s^{-n-1}\langle N\rangle_0, \label{eq:annuli3}
\end{align}
for some functions $g_k\in\cG$.

The previous inequalities together with \eqref{eq:deriv-annuli} yield
\begin{align*}
\lan D\Phi_s(t)\ran_t\  &\le (\kappa-v) s^{-1} \big\langle \mathrm{d}\Gamma\big ( f^\prime(|x|^+_{ts}) - g^\prime(|x|^-_{ts}) \big) \big\rangle_t \\
&\quad + C_{f,g,n} \sum_{k=2}^n \kappa^{(k)}_J s^{-k} \big\langle \mathrm{d}\Gamma\big ( \tilde f^\prime_k(|x|^+_{ts})-\tilde g^\prime_k(|x|_{ts}) \big) \big\rangle_t \\
&\quad +C_{f,g,n} \kappa_J^{(n+1)} s^{-n-1}\langle N\rangle_0.
\end{align*}
Hence, Eqs. \eqref{dt-Heis} and \eqref{eq-basic} and the definition $\Phi_s(t)= \mathrm{d}\Gamma\big ( f(|x|^+_{ts})g(|x|^-_{ts}) \big )$ give  
\begin{align} \label{propag-est2-annuli} 
&\big\lan \mathrm{d}\Gamma\big ( f(|x|^+_{ts})g(|x|^-_{ts}) \big ) \big\ran_t+(v-\kappa)s^{-1}\int_0^t \big\lan \mathrm{d}\Gamma\big ( f'(|x|^+_{rs}) - g'(|x|^-_{rs}) \big )\big \ran_r dr\notag\\
& \le \big\lan \mathrm{d}\Gamma\big ( f(|x|^+_{0s} g(|x|^-_{0s}))  \big ) \big\ran_0+C_{f,g,n} \sum_{k=2}^n s^{-k}\int_0^t \big\lan  \mathrm{d}\Gamma\big ( \tilde f'_k(|x|^+_{rs})-\tilde g'_k(|x|^-_{rs}) \big ) \big\ran_r dr\notag\\
& \quad + C_{f,g,n} t s^{-n-1}\langle N\rangle_0. 
\end{align}
Since $\kappa < v$, we can then iterate the process and conclude as in the proof of Proposition \ref{prop:propag-est1}.
\end{proof}

\begin{proof}[End of the proof of Theorem \ref{thm:max-vel-est-annuli}]
Since $f\in \cE$ and $g\in\cG$, we have
\begin{align}\label{local-est4-annuli}
\big \lan \mathrm{d}\Gamma\big ( f(|x|^+_{0s}) g(|x|^-_{0s}) \big ) \big\ran_0 \le \big\langle \mathrm{d}\Gamma( \chi_{|x|>b}\chi_{|x|<2\rho-b} )\rangle_0 .\end{align}

Next, retaining the first term in \eqref{propag-est2-annuli}, dropping the second one and  using Proposition \ref{prop:propag-est1-annuli}, we conclude as in the proof of Theorem \ref{thm:max-vel-est} that
\begin{align} 
\big\lan \mathrm{d}\Gamma\big ( f(|x|^+_{ts}) g(|x|^-_{ts}) \big ) \big\ran_t &\le (1+C_{f,g,c,n}s^{-1}) \langle \mathrm{d}\Gamma( \chi_{b<|x|<2\rho-b})\rangle_0 \notag \\
&\quad+  C_{f,g,c,n} s^{-n}  \langle N\rangle_0, \label{propag-est4-annuli} 
\end{align}
for any $0\le t\le s\le c^{-1}(\rho-b)$ and $n\le p -1$. Now, for any $f\in \cE $, we have $f(\frac{\cdot-b-vt}{s})=1$ on 
$[ b+vt +(c-v) s , \infty)$. 
For $\beta\rho \ge b + cs$ and $s\ge t$ , we have $[\beta\rho,\infty) \subset [ b+vt +(c-v) s , \infty)$.
Hence 
\begin{equation*}
\chi_{|x|>\beta\rho}\le f(|x|^+_{ts}).
\end{equation*}
Likewise,
\begin{equation*}
\chi_{|x|<(2-\beta)\rho}\le g(|x|^-_{ts}).
\end{equation*}
Therefore, choosing $s=c^{-1}(\beta\rho-b)$, we conclude that, for $\beta\rho \ge b + ct$,
\begin{align*}
\big\lan \mathrm{d}\Gamma( \chi_{\beta\rho<|x|<(2-\beta)\rho}) \big\ran_t &\le \big\lan \mathrm{d}\Gamma\big ( f(|x|^+_{ts})g(|x|^-_{ts}) \big ) \big\ran_t \\
&\le (1+C_{f,g,c,n}s^{-1}) \langle \mathrm{d}\Gamma( \chi_{b<|x|<2\rho-b})\rangle_0 +  C_{f,g,c,n} s^{-n}  \langle N\rangle_0.
\end{align*}
This concludes the proof.
\end{proof}

\section{Proof of Theorem \ref{thm:MVE-SM2}}
\label{sec:thm3}

In this section, we explain how to modify the proof of Theorem \ref{thm:MVE-SM} in order to obtain Theorem \ref{thm:MVE-SM2}. We consider $\varphi\in\mathcal{D}(N^{\frac12(1+\nu_A+\nu_B)})\subset\cF$ and estimate $\langle\varphi,e^{-itH}Ae^{itH}B\varphi\rangle$ by using  \eqref{eq:new1} and \eqref{eq:atAB_31} as in the proof of Theorem \ref{thm:MVE-SM}. To shorten formulas, we use in this proof the notation
\begin{equation*}
\bar N := N+1.
\end{equation*}

Instead of \eqref{eq:P1_0},  the third term of \eqref{eq:atAB_31} is estimated as:
\begin{align}
\big\|N^{\frac12}e^{-irH}A e^{itH}B\varphi\big\|&=\big\|N^{\frac12}A e^{itH}B\varphi\big\|\notag\\
&\le\|A\|_{1}\big\|\bar N^{\frac12(1+\nu_A)} e^{itH}B\varphi\big\|\notag\\
&=\|A\|_{1}\big\| \bar N^{\frac12(1+\nu_A)} B\varphi\big\|\notag\\
&\le\|A\|_{1} \|B\|_{1+\nu_A}\big\|\bar N^{\frac12(1+\nu_A+\nu_B)}\varphi\big\|, \label{eq:P1}
\end{align}
where we used the fact that $N$ commutes with $H$ and the unitarity of $e^{itH}$ in the equalities, and \eqref{eq:compl} in the inequalities.

Instead of \eqref{eq:l2_1}, the last term in \eqref{eq:atAB_31} is estimated using \eqref{eq:P1}, yielding
\begin{align} 
&\big\| (N_{\alpha,\rho}^{(1)} )^{\frac12} e^{-irH} A e^{itH}B\varphi\big \|\notag\\
& \le C_{c,p,\delta_0} \Big( \big\| (N_{\alpha,\rho}^{(2)} )^{\frac12} A e^{itH}B\varphi\big\| \notag\\
&\qquad\qquad+(\rho-b)^{-\frac{p-1}{2}} \|A\|_{1} \|B\|_{1+\nu_A}\big\|\bar N^{\frac12(1+\nu_A+\nu_B)}\varphi\big\|\Big) . \label{eq:l2_1_app}
\end{align}

To control the first term in the right-hand-side of \eqref{eq:l2_1_app}, we uses that $A$ commutes with $N_{\alpha,\rho}^{(2)}$ as in \eqref{eq:comm_A_N} and next use \eqref{eq:compl}, which gives
\begin{align}
&\big\| (N_{\alpha,\rho}^{(2)} )^{\frac12} A e^{itH}B\varphi\big\| \le\|A\|_{0}\big\| (N_{\alpha,\rho}^{(2)} )^{\frac12} e^{itH} \bar N^{\frac{\nu_A}{2}}B\varphi\big\|,
\end{align}
where we also used that $N$ commutes with $N_{\alpha,\rho}^{(2)}$ and $H$. Applying Theorem \ref{thm:max-vel-est-annuli} as in \eqref{eq:l2_2}, we obtain
\begin{align}
&\big\| (N_{\alpha,\rho}^{(2)} )^{\frac12} e^{itH}\bar N^{\frac{\nu_A}{2}}B\varphi\big\|\notag\\
& \le C_{c,p,\delta_0} \Big(\big\| (N_{\alpha,\rho}^{(3)} )^{\frac12} \bar N^{\frac{\nu_A}{2}}B\varphi\big\| + (\rho-b)^{-\frac{p-1}{2}}\big\| \bar N^{\frac12(1+\nu_A)} B\varphi\big\|\Big) . \label{eq:l2_2_app}
\end{align}
As above, we have
\begin{equation}
\big\| \bar N^{\frac12(1+\nu_A)} B\varphi\big\|\le\|B\|_{1+\nu_A}\big\|\bar N^{\frac12(1+\nu_A+\nu_B)}\varphi\big\|,
\end{equation}
and using that $B$ commutes with $N_{\alpha,\rho}^{(3)}$ as in \eqref{eq:l2_2_0}, we can write
\begin{align}
\big\| (N_{\alpha,\rho}^{(3)} )^{\frac12} \bar N^{\frac{\nu_A}{2}}B\varphi\big\| &=\big\| \bar N^{\frac{\nu_A}{2}}B(N_{\alpha,\rho}^{(3)} )^{\frac12} \varphi\big\| \notag\\
&\le \|B\|_{\nu_A} \big\| (N_{\alpha,\rho}^{(3)} )^{\frac12} \bar N^{\frac12(\nu_A+\nu_B)} \varphi\big\|.
\end{align}
Moreover, recalling that $m_\rho(x)=\min(|x|,2\rho-|x|)$, we have $m_\rho(x)\ge(1-\alpha)\rho+b$ for all $x\in\mathcal{C}_{b+(1-\alpha)\rho,(1+\alpha)\rho-b}$ and hence, since $N_{\alpha,\rho}^{(3)}=N_{\mathcal{C}_{b+(1-\alpha)\rho,(1+\alpha)\rho-b}}$,
\begin{align}
&\big\| (N_{\alpha,\rho}^{(3)} )^{\frac12}\bar N^{\frac12(\nu_A+\nu_B)} \varphi\big\|\notag\\
&\le\big((1-\alpha)\rho+b\big)^{-\frac{n}{2}}\big\| \bar N^{\frac12(\nu_A+\nu_B)}\mathrm{d}\Gamma\big(m_\rho^n\chi_{\mathcal{C}_{b+(1-\alpha)\rho,(1+\alpha)\rho-b}}\big)^{\frac12} \varphi\big\| \notag\\
&\le\big((1-\alpha)\rho+b\big)^{-\frac{n}{2}}\big\| \bar N^{\frac12(\nu_A+\nu_B)}(G_\rho^{(n)})^{\frac12} \varphi\big\|, \label{eq:l2_4}
\end{align}
for all $n\in\mathbb{N}_0$, where we used that $\alpha<1$ in the second inequality and set
\begin{equation*}
G_\rho^{(n)}:=\mathrm{d}\Gamma\big(m_\rho^n\chi_{\mathcal{C}_{b,2\rho-b}}\big).
\end{equation*}
 Putting together Eq. \eqref{eq:l2_1_app}--\eqref{eq:l2_4}, we have proven that
\begin{align}
&\big\| (N_{\alpha,\rho}^{(1)} )^{\frac12} e^{-irH} A e^{itH}B\varphi\big \|\notag\\
&\le C_{c,p,\delta_0}(\rho-b)^{-\frac{p-1}{2}} \vertiii{A}_{\nu_B} \vertiii{B}_{\nu_A} \big\| \bar N^{\frac12(1+\nu_A+\nu_B)}\varphi\big\| \notag \\
&\quad+C_{c,p}\big((1-\alpha)\rho+b\big)^{-\frac{p-1}{2}} \vertiii{A}_{\nu_B} \vertiii{B}_{\nu_A}\big\| \bar N^{\frac12(\nu_A+\nu_B)}(G_\rho^{(p-1)})^{\frac12} \varphi\big\|, \label{eq:estim_2_app}
\end{align}
where, recall, $\vertiii{A}_{\nu}=\max_{0\le n\le\nu+1}\|A\|_n$.

The term $\| (N_{\alpha,\rho}^{(1)} )^{\frac12} U_\rho^*e^{i(t-r)\tilde H} U_\rho\varphi\|$ in \eqref{eq:atAB_31} is estimated in the same way as in \eqref{eq:estim_3_0}, using in addition the function $m_\rho$ as in the previous equation. This gives
\begin{align}
&\big\| (N_{\alpha,\rho}^{(1)} )^{\frac12} U_\rho^*e^{i(t-r)\tilde H} U_\rho\varphi \big \| \notag\\
&\le C_{c,p,\delta_0}\big((1-\alpha)\rho+b\big)^{-\frac{p-1}{2}}\Big(\big\|N^{\frac12}\varphi\big\|+\big\| (G_\rho^{(p-1)})^{\frac12} \varphi\big\|\Big). \label{eq:estim_3}
\end{align}

Arguing as in \eqref{eq:atAB_33_0}, we then deduce from the previous equations that 
\begin{align}
&\big|\big\langle\varphi, \mathrm{Rem}_1 \,\varphi\big\rangle\big| \le C_{c,p,\delta_0,b} \vertiii{A}_{\nu_B} \vertiii{B}_{\nu_A} t (\rho-b)^{-p+1} \notag  \\
&\quad \times \big ( \big\langle \varphi , \bar N^{1+\nu_A+\nu_B} \varphi \big\rangle + \big\langle\varphi, \bar N^{\nu_A+\nu_B}(G_\rho^{(p-1)})\varphi\big\rangle \big) .\label{eq:atAB_33}
\end{align}
One proceeds in the same way to estimate the remainder term $\mathrm{Rem}_2$ in \eqref{eq:new1_1} and then concludes exactly as in the proof of Theorem \ref{thm:MVE-SM}. The proof of \eqref{eq:general_loc} follows in the same way as the proof of Theorem \ref{thm:localapprox}, taking $B=\one$.
\qed

\begin{proof}[Proof of \eqref{max-vel-estSM_2nd}]
The proof of \eqref{max-vel-estSM_2nd} is easier. The only difference comes from the last term in \eqref{eq:atAB_31}, which is estimated using that 
\begin{equation*}
\big\| (N_{\alpha,\rho}^{(1)} )^{\frac12} \bar N^{-\frac12} \big\| \le C.
\end{equation*}
Together with \eqref{eq:P1}, this gives
\begin{align*} 
&\big\| (N_{\alpha,\rho}^{(1)} )^{\frac12} e^{-irH} A e^{itH}B\varphi\big \| \le C \|A\|_{1} \|B\|_{1+\nu_A}\big\|\bar N^{\frac12(1+\nu_A+\nu_B)}\varphi\big\| , 
\end{align*}
instead of \eqref{eq:estim_2_app}. The rest of the proof is identical.
\end{proof}

\section{Proof of the bound on quantum state transfer (Corollary \ref{cor:qst})}\label{app:qst}

We follow the line of argumentation from \cite{EppWhal} and adapt it to bosons. Two important differences are that (a) we work with the fidelity throughout since it is better suited for form bound, i.e., we avoid the Fuchs-van der Graaf inequality, and (b) we replace the use of formula (1) by an application of \eqref{eq:general_loc}.

\begin{proof}[Proof of Corollary \ref{cor:qst}]
Note that the localized operator $(\alpha_t(A))_\rho\equiv [A_t]_\rho$ introduced in Theorem \ref{thm:MVE-SM2} (together with Remark \ref{rmk:extensions} (iii) since $A$ is unitary) is localized in $\mathcal B(0,\rho)\subset Y^c$ and so conjugation by it does not affect $\mathrm{Tr}_{Y^c}$, leading to
$$
\begin{aligned}
F\left( \mathrm{Tr}_{Y^c}(\gamma_t), \mathrm{Tr}_{Y^c}A_t\gamma_t A_t^*)\right)
&=
F\left(\mathrm{Tr}_{Y^c} \left( [A_t]_\rho\gamma_t [A_t]^*_\rho\right), \mathrm{Tr}_{Y^c}\left( A_t\gamma_t A_t^*\right)\right)\\
&\geq
F\left( [A_t]_\rho\gamma_t [A_t]^*_\rho, A_t\gamma_t A_t^*\right)
\end{aligned}
$$
where the estimate holds by the data processing inequality for the fidelity. 

Note that $[A_t]_\rho$ is unitary and satisfies $[A_t]_\rho^*=[A_t^*]_\rho$, see \eqref{eq:def_alpha_trho}. Since $\gamma$ is pure, $\gamma_t=\vert \phi_t\rangle\langle \phi_t\vert$ is also pure and the fidelity becomes
$$
F\left( [A_t]_\rho\gamma_t [A_t]^*_\rho, A_t\gamma_t A_t^*\right)
=|\langle [A_t]_\rho\phi_t,A_t\phi_t\rangle|
\geq 1- |\langle \phi_t,(A_t-[A_t]_\rho)\phi_t\rangle|.
$$

We recall \eqref{eq:general_loc} which in the present notation says that for any $\varphi$,
\begin{align}
\big|\big \langle \varphi ,\left(A_t-[A_t]_\rho\right)\varphi \big \rangle \big | 
\le C t (\rho-b)^{1-p} \vertiii{A}_0 \big\langle \varphi , M_\rho \varphi \big\rangle .
\end{align}
This proves Corollary \ref{cor:qst}.
\end{proof}

\medskip

\begin{small}
\noindent \textbf{Acknowledgements.} The research of IMS was supported in part by NSERC grant 7901. We are grateful to the anonymous referees for their constructive remarks.
\end{small}


\begin{thebibliography}{10}

\bibitem{APSS} J. Arbunich, F. Pusateri, I.M. Sigal, A. Soffer, Maximal Speed of Quantum Propagation, Lett. Math. Phys. Volume 111, Issue 3, 1-16 (2021). 

\bibitem{BMNS}
S.~Bachmann, S.~Michalakis, B.~Nachtergaele, and R.~Sims, Automorphic equivalence within gapped phases of quantum lattice systems, Communications in Mathematical Physics, 309 (2012), no.\ 3, 835-871.

\bibitem{BHV}
S.~Bravyi, M.B.~Hastings, and F.~Verstraete, F. Lieb-Robinson bounds and the generation of correlations and topological quantum order. Physical review letters, 97(5), 050401 (2006).

\bibitem{BoFauSig}  J.-F. Bony,  J. Faupin,  I.M. Sigal : 
Maximal velocity of photons in non-relativistic QED. Adv. Math. 231, 3054--3078 (2012).



\bibitem{CrSerEis} M. Cramer, A. Serafini, and J. Eisert, Locality of dynamics in general harmonic quantum systems, arXiv:0803.0890 (2008).



\bibitem{DerGer1} J. Derezi\'nski and C. G\'erard :
{\it Scattering Theory of Classical and Quantum $N$-Particle Systems}.
Springer-Verlag:  Berlin, 1997.



\bibitem{DerGer2} J. Derezi\'{n}ski and C. G\'erard,  \emph{Asymptotic
completeness in quantum field theory. Massive Pauli-Fierz
Hamiltonians}, Rev. Math. Phys., 11, (1999), 383--450.


\bibitem{ElMaNayakYao} D. V. Else, F. Machado, Ch. Nayak, and N. Y. Yao:  Improved Lieb-Robinson bound for many-body Hamiltonians with power-law interactions.
Phys. Rev. A 101, 022333, 2020. 


\bibitem{EppWhal} J. M. Epstein, K. B. Whaley, Quantum speed limits for quantum-information-processing task, Phys. Rev. A 95, 042314 (2017).



\bibitem{FLS} J. Faupin, M. Lemm and I.M.Sigal: Maximal speed for macroscopic particle transport in the Bose-Hubbard model, Phys. Rev. Lett. 128, 150602 (2022).



\bibitem{FaupSig}   
J. Faupin,  I.M.Sigal :  On Rayleigh scattering in non-relativistic quantum electrodynamics.  Commun. Math. Phys. 328, 1199--1254 (2014).




\bibitem{FossEtAl} 
M. Foss-Feig, Zhe-Xuan Gong, Ch. W. Clark, and A. V. Gorshkov:  Nearly-linear light cones in long-range interacting quantum systems. Phys. Rev. Lett. 114, 157201 (2015). 



\bibitem{FrGrSchl2}
J.~Fr\"ohlich, M.~Griesemer and B.~Schlein,
\emph{Asymptotic completeness for Rayleigh scattering}, Ann. Henri Poincar\'e, 3, (2002), 107--170.




\bibitem{GebNachReSims}
M. Gebert, B. Nachtergaele, J. Reschke, R. Sims:
Lieb-Robinson bounds and strongly continuous dynamics for a class of many-body fermion systems in $\R^d$.
Annales Henri Poincar\'e,  21(11):3609-3637 (2020). 

		\bibitem{HastingsAL}
M.B.~Hastings, \textit{An area law for one-dimensional quantum systems}, J.\ Stat.\ Mech.: Theor.\ Exper.\ \textbf{2007} (2007), P08024

	\bibitem{Hastings}
M.B.~Hastings,\emph{{Lieb-Schultz-Mattis in higher dimensions}} Phys.\ Rev.\ B \textbf{69} (2004),104431	


 
\bibitem{HeSk} I. Herbst and E. Skibsted : Free channel Fourier transform in the long-range N-body problem. J. d'Analyse Math. 65 (1995) 297--332.


\bibitem{HuSp}  M. H\"ubner and  H. Spohn, \emph{Radiative decay: nonperturbative approaches}, Rev. Math. Phys., 7, (1995), 363--387.


\bibitem{KuwSaito} 
T. Kuwahara and K. Saito, Lieb-Robinson bound and almost-linear light-cone in interacting boson systems, Phys. Rev. Lett. 127, 070403, (2021).


\bibitem{LR} E. H. Lieb and D. Robinson: The finite group velocity of quantum spin systems. Commun. math. Phys. 28, 251-257 (1972).


\bibitem{MatKoNaka} T. Matsuta, T. Koma and S. Nakamura: Improving the Lieb-Robinson bound for long-range interactions.
 Annales Henri Poincar\'e volume 18, 519--528 (2017).
 
\bibitem{NachRazSchlSim}  B. Nachtergaele, H. Raz, B. Schlein, and R. Sims, Lieb-Robinson bounds for harmonic and anharmonic lattice systems, Commun. Math. Phys. 286, 1073 (2009).
 


\bibitem{NachSim} B. Nachtergaele and R. Sims: Much ado about something
why Lieb-Robinson bounds are useful. arXiv:1102.0835

\bibitem{NachSim04} B.~Nachtergaele and R.~Sims, \textit{A Multi-Dimensional
Lieb-Schultz-Mattis Theorem}, Comm.\ Math.\ Phys.\ \textbf{276} (2007), 437


\bibitem{NSY}
B.~Nachtergaele, R.~Sims, and A.~Young, \textit{Quasi-locality bounds for quantum lattice systems. I. Lieb-Robinson bounds, quasi-local maps, and spectral flow automorphisms} J.\ Math.\ Phys.\ \textbf{60} (2019), no.\ 6, 061101

\bibitem{NVZ}
B.~Nachtergaele, A.~Vershynina, and V.~Zagrebnov, \textit{Lieb-Robinson Bounds and Existence of the Thermodynamic Limit for a Class of Irreversible Quantum Dynamics}, Contemporary Mathematics \textbf{552}, 161-175 (2011).

\bibitem{Pou}D. Poulin, \textit{Lieb-Robinson bound and locality for general Markovian quantum dynamics}. Phys. Rev. Lett. 104 (2010) 190401. 


\bibitem{SHOE} 
N. Schuch, S. K. Harrison, T. J. Osborne, and J. Eisert, Information propagation for interacting-particle systems. Physical Review A 84, 032309 (2011).


\bibitem{SigSof} I.M. Sigal and A. Soffer : Local decay and propagation estimates for time-dependent and time-independent Hamiltonians. 
Preprint, Princeton Univ. (1988)  http://www.math.toronto.edu/sigal/publications/SigSofVelBnd.pdf.

\bibitem{Skib} E. Skibsted : Propagation estimates for N-body Schr\"odinger operators. Comm. Math. Phys. 142 (1992) 67--98.


\bibitem{WaHazz} Zhiyuan Wang and K. 
 R.A. Hazzard 
: Tightening the Lieb-Robinson Bound in Locally Interacting Systems. PRX Quantum 1, 010303 (2020)

\bibitem{YinLuc}
		C.~Yin and A.~Lucas, \textit{Finite speed of quantum information in models of interacting bosons at finite density}, https://arxiv.org/pdf/2106.09726.pdf
  
\end{thebibliography}
\end{document}